\documentclass[12pt]{amsart}
\reversemarginpar
\usepackage{amsmath}
\usepackage{amssymb}
\usepackage{graphicx}
\usepackage{graphicx}
\usepackage{dcolumn}
\usepackage{bm}
\usepackage{amsfonts}
\usepackage{latexsym}
\usepackage{pdfsync}
    \usepackage{fullpage} 
\usepackage{subfigure}
\usepackage{colordvi}
\usepackage{color}

\begin{document}
\newcommand{\commentout}[1]{}

\newcommand{\nwc}{\newcommand}
\newcommand{\bz}{{\mathbf z}}
\newcommand{\sqk}{\sqrt{\ks}}
\newcommand{\sqkone}{\sqrt{|\ks_1|}}
\newcommand{\sqktwo}{\sqrt{|\ks_2|}}
\newcommand{\invsqkone}{|\ks_1|^{-1/2}}
\newcommand{\invsqktwo}{|\ks_2|^{-1/2}}
\newcommand{\partz}{\frac{\partial}{\partial z}}
\newcommand{\grady}{\nabla_{\by}}
\newcommand{\gradp}{\nabla_{\bp}}
\newcommand{\invf}{\cF^{-1}_2}
\newcommand{\myphi}{\Phi_{(\eta,\rho)}}
\newcommand{\minrg}{|\min{(\rho,\gamma^{-1})}|}
\newcommand{\al}{\alpha}
\newcommand{\xvec}{\vec{\mathbf x}}
\newcommand{\kvec}{{\vec{\mathbf k}}}
\newcommand{\lt}{\left}
\newcommand{\ksq}{\sqrt{\ks}}
\newcommand{\rt}{\right}
\nwc{\bG}{{\bf G}}
\newcommand{\ga}{\gamma}
\newcommand{\vas}{\varepsilon}
\newcommand{\lan}{\left\langle}
\newcommand{\ran}{\right\rangle}
\newcommand{\tvas}{{W_z^\vas}}
\newcommand{\psiep}{{W_z^\vas}}
\newcommand{\wep}{{W^\vas}}
\newcommand{\weptil}{{\tilde{W}^\vas}}
\newcommand{\wepz}{{W_z^\vas}}
\newcommand{\weps}{{W_s^\ep}}
\newcommand{\wepsp}{{W_s^{\ep'}}}
\newcommand{\wepzp}{{W_z^{\vas'}}}
\newcommand{\wepztil}{{\tilde{W}_z^\vas}}
\newcommand{\vvas}{{\tilde{\ml L}_z^\vas}}
\newcommand{\veptil}{{\tilde{\ml L}_z^\vas}}
\newcommand{\vep}{{{ V}_z^\vas}}
\newcommand{\cvc}{{{\ml L}^{\ep*}_z}}
\newcommand{\cvcp}{{{\ml L}^{\ep*'}_z}}
\newcommand{\cvp}{{{\ml L}^{\ep*'}_z}}
\newcommand{\cvtil}{{\tilde{\ml L}^{\ep*}_z}}
\newcommand{\cvtilp}{{\tilde{\ml L}^{\ep*'}_z}}
\newcommand{\vtil}{{\tilde{V}^\ep_z}}
\newcommand{\ktil}{\tilde{K}}
\newcommand{\n}{\nabla}
\newcommand{\tkappa}{\tilde\kappa}
\newcommand{\ks}{{\omega}}
\newcommand{\mbx}{\mb x}
\newcommand{\br}{\mb r}
\nwc{\bH}{{\mb H}}
\newcommand{\bu}{\mathbf u}
\nwc{\bxp}{{{\mathbf x}}}
\nwc{\byp}{{{\mathbf y}}}
\newcommand{\bD}{\mathbf D}
\nwc{\bh}{\mathbf h}
\newcommand{\bB}{\mathbf B}
\newcommand{\bb}{\mathbf b}
\newcommand{\bC}{\mathbf C}
\nwc{\cO}{\mathcal  O}
\newcommand{\bp}{\mathbf p}
\newcommand{\bq}{\mathbf q}
\newcommand{\by}{\mathbf y}
\nwc{\bP}{\mathbf P}
\nwc{\bs}{\mathbf s}
\nwc{\bX}{\mathbf X}
\newcommand{\pdg}{\bp\cdot\nabla}
\newcommand{\pdgx}{\bp\cdot\nabla_\bx}
\newcommand{\one}{1\hspace{-4.4pt}1}
\newcommand{\corr}{r_{\eta,\rho}}
\newcommand{\rinf}{r_{\eta,\infty}}
\newcommand{\rzero}{r_{0,\rho}}
\newcommand{\rzeroinf}{r_{0,\infty}}
\nwc{\om}{\omega}
\nwc{\Gp}{{G_{\rm par}}}
\nwc{\nwt}{\newtheorem}
\nwc{\xp}{{x^{\perp}}}
\nwc{\yp}{{y^{\perp}}}
\nwt{remark}{Remark}
\nwt{definition}{Definition} 
\nwc{\bd}{{\mb d}}
\nwc{\ba}{{\mb a}}
\nwc{\mbe}{{\mb e}}
\nwc{\bal}{\begin{align}}
\nwc{\bea}{\begin{eqnarray}}
\nwc{\beq}{\begin{eqnarray}}
\nwc{\bean}{\begin{eqnarray*}}
\nwc{\beqn}{\begin{eqnarray*}}
\nwc{\beqast}{\begin{eqnarray*}}

\nwc{\eal}{\end{align}}
\nwc{\eea}{\end{eqnarray}}
\nwc{\eeq}{\end{eqnarray}}
\nwc{\eean}{\end{eqnarray*}}
\nwc{\eeqn}{\end{eqnarray*}}
\nwc{\eeqast}{\end{eqnarray*}}

\nwc{\ep}{\varepsilon}
\nwc{\eps}{\varepsilon}
\nwc{\ept}{\epsilon}
\nwc{\vrho}{\varrho}
\nwc{\orho}{\bar\varrho}
\nwc{\ou}{\bar u}
\nwc{\vpsi}{\varpsi}
\nwc{\lamb}{\lambda}
\nwc{\Var}{{\rm Var}}

\nwt{cor}{Corollary}
\nwt{proposition}{Proposition}
\nwt{corollary}{Corollary}
\nwt{theorem}{Theorem}
\nwt{summary}{Summary}
\nwt{lemma}{Lemma}
\nwc{\nn}{\nonumber}
\nwc{\mf}{\mathbf}
\nwc{\mb}{\mathbf}
\nwc{\ml}{\mathcal}
\nwc{\bj}{{\mb j}}
\nwc{\bA}{{\mb A}}
\nwc{\IA}{\mathbb{A}} 
\nwc{\bi}{\mathbf i}
\nwc{\bo}{\mathbf o}
\nwc{\IB}{\mathbb{B}}
\nwc{\IC}{\mathbb{C}} 
\nwc{\ID}{\mathbb{D}} 
\nwc{\IM}{\mathbb{M}} 
\nwc{\IP}{\mathbb{P}} 
\nwc{\bI}{\mathbf{I}} 
\nwc{\IE}{\mathbb{E}} 
\nwc{\IF}{\mathbb{F}} 
\nwc{\IG}{\mathbb{G}} 
\nwc{\IN}{\mathbb{N}} 
\nwc{\IQ}{\mathbb{Q}} 
\nwc{\IR}{\mathbb{R}} 
\nwc{\IT}{\mathbb{T}} 
\nwc{\IZ}{\mathbb{Z}} 
\nwc{\II}{\mathbb{I}} 

\nwc{\cE}{{\ml E}}
\nwc{\cP}{{\ml P}}
\nwc{\cQ}{{\ml Q}}
\nwc{\cL}{{\ml L}}
\nwc{\cX}{{\ml X}}
\nwc{\cW}{{\ml W}}
\nwc{\cZ}{{\ml Z}}
\nwc{\cR}{{\ml R}}
\nwc{\cV}{{\ml V}}
\nwc{\cT}{{\ml T}}
\nwc{\crV}{{\ml L}_{(\delta,\rho)}}
\nwc{\cC}{{\ml C}}
\nwc{\cA}{{\ml A}}
\nwc{\cK}{{\ml K}}
\nwc{\cB}{{\ml B}}
\nwc{\cD}{{\ml D}}
\nwc{\cF}{{\ml F}}
\nwc{\cS}{{\ml S}}
\nwc{\cM}{{\ml M}}
\nwc{\cG}{{\ml G}}
\nwc{\cH}{{\ml H}}
\nwc{\cN}{{\ml N}}
\nwc{\bk}{{\mb k}}
\nwc{\bT}{{\mb T}}
\nwc{\cbz}{\overline{\cB}_z}
\nwc{\supp}{{\hbox{\rm supp}}}
\nwc{\fR}{\mathfrak{R}}
\nwc{\bY}{\mathbf Y}
\newcommand{\mbr}{\mb r}
\nwc{\pft}{\cF^{-1}_2}
\nwc{\bU}{{\mb U}}
\nwc{\bPhi}{{\mb \Phi}}
\nwc{\bPsi}{{\mb \Psi}}

\commentout{
\newenvironment{proof}[1][Proof]{\begin{trivlist}
\item[\hskip \labelsep {\bfseries #1}]}{\end{trivlist}}
\newenvironment{definition}[1][Definition]{\begin{trivlist}
\item[\hskip \labelsep {\bfseries #1}]}{\end{trivlist}}
\newenvironment{example}[1][Example]{\begin{trivlist}
\item[\hskip \labelsep {\bfseries #1}]}{\end{trivlist}}
\newenvironment{remark}[1][Remark]{\begin{trivlist}
\item[\hskip \labelsep {\bfseries #1}]}{\end{trivlist}}
}

\newcommand{\CC}{\mathbb{C}}
\newcommand{\hj}{\hat{J}}
\newcommand{\tj}{\tilde{J}}
\newcommand{\xmax}{x_{\text{max}}}
\newcommand{\xmin}{x_{\text{min}}}
\newcommand{\bx}{{x}}
\newcommand{\be}{\bar{e}}
\newcommand{\emax}{c_{\text{max}}}
\nwc{\red}{\color{red}}
\nwc{\blue}{\color{blue}}
\nwc{\green}{\color{green}}
\newcommand{\RR}{\mathbb{R}}
\nwc{\suppx}{{\hbox{supp}(\mathbf x)}}
\nwc{\bn}{\mathbf{n}}

\title{Coherence-Pattern  Guided Compressive Sensing  with Unresolved Grids}
\author{Albert C.  Fannjiang*}
  \address{
   Department of Mathematics,
    University of California, Davis, CA 95616-8633}
  \thanks{
 *Corresponding author:  fannjiang@math.ucdavis.edu. 
The research supported in part by NSF Grant DMS 0908535}
\author{ Wenjing Liao }
\maketitle

\begin{abstract}

Highly coherent sensing matrices arise  in discretization of 
continuum  imaging problems 
such as radar  and medical imaging when the grid spacing
is below the Rayleigh threshold.

Algorithms based on techniques of band exclusion (BE) and local optimization (LO) are proposed
to deal with such coherent sensing matrices. These techniques
are embedded in the existing compressed sensing algorithms
such as Orthogonal Matching Pursuit (OMP),
Subspace Pursuit (SP), Iterative Hard Thresholding (IHT),
Basis Pursuit (BP) and Lasso,  and result in the modified algorithms
BLOOMP, BLOSP, BLOIHT, BP-BLOT and Lasso-BLOT, respectively.

Under appropriate conditions, it is proved that BLOOMP
can reconstruct sparse, widely separated objects 
up to one Rayleigh length in the Bottleneck distance
{\em independent}   of the grid spacing. One of
the most distinguishing attributes of BLOOMP
is its capability of dealing with large dynamic ranges. 

The BLO-based algorithms are systematically tested 
with respect to  four
performance metrics: dynamic range, noise stability, 
sparsity  and resolution. 
With respect to  dynamic range and noise stability,
BLOOMP is the best performer. 
With respect to sparsity, BLOOMP is the best
performer for high dynamic range while for dynamic range
near unity BP-BLOT and Lasso-BLOT with the optimized regularization parameter have the best performance.  In the noiseless case, 
BP-BLOT has the highest resolving power  up to certain dynamic range. 

 The algorithms
BLOSP and  BLOIHT  are good alternatives  to 
 BLOOMP and BP/Lasso-BLOT: they are faster
than both BLOOMP and BP/Lasso-BLOT  and shares, to a lesser degree, BLOOMP's amazing attribute with respect to dynamic range.   

Detailed comparisons with existing algorithms such
as Spectral Iterative Hard Thresholding (SIHT) and the frame-adapted
BP are given. 
\end{abstract}

\section{Introduction}
Reconstruction  of a high-dimensional sparse signal from sparse linear measurements is
a fundamental problem relevant to  
imaging, inverse problems and  signal processing.  

Consider, for example,  the problem  of spectral estimation  in signal processing.  Let the uncontaminated signal $y(t)$ be a  linear
combinations of 
$s$ time-harmonic components
\[
\{e^{-i2\pi \om_j t}: j=1,...,s\}, 
\] 
namely
\beq
\label{50}
y(t)=\sum_{j=1}^{s} c_j e^{-i2\pi \om_j t}
\eeq
where $c_j$  are the amplitudes. Suppose that $y(t)$ is contaiminated by noise $n(t)$ and the received signal is 
 \beq
 \label{40}
 b(t)=y(t)+n(t). 
 \eeq
The task is to find out the frequencies $\Omega=\{\om_j\}$ and
the amplitudes $\{c_j\}$
by sampling $b(t)$ at discrete times.  
 
A standard approach to spectral estimation  is to turn the problem  into 
  the linear inversion problem as follows.
  To fix the idea, let  $t_k, k=1,...,N$ be the sample times 
  in the unit interval $[0,1]$. Set $\bb=(b(t_k))\in \IC^N$ to be the data vector. 
We approximate the frequencies by the unknown closest subset of cardinality $s$  of a regular grid $\cG=
\{p_1,\ldots, p_M\}$   
and write the corresponding amplitudes as $\mbx=(x_j)\in \IC^M$
where the components of $\mbx$ equal the amplitudes $\{c_j\}$
whenever the grid points are the {\em nearest}  grid points to the 
frequencies $\{\om_j\}$  and zero otherwise. 
Typically the number of approximating grid points is far greater than
the number of frequencies, i.e. $M\gg s$.

Let the measurement matrix be 
\beq
\label{2}
\bA=\begin{bmatrix}
     \ba_{1} & \ldots & \ba_{M}
  \end{bmatrix}\in \IC^{N\times M}
  \eeq
  with
  \beq
\label{3}
 \ba_j={1\over \sqrt{N}}\lt( e^{-i2\pi t_kp_j}\rt) \in \IC^N,\quad j=1,...,M.
  \eeq
We cast  the spectral estimation problem into the form
\beq\label{1}
    \bA \mbx + \mbe = \bb
    \label{linearsystem}
\eeq
where the error vector 
 $\mbe=(e_k)\in \IC^N$ is the sum of  the external noise
 $\bn=(n(t_k))$ and the discretization or gridding  error $\bd=(\delta_k) \in \IC^N $
 due to approximating 
 the frequencies by the grid points in $\cG$. 
By definition, the gridding  error is given by 
\beq
\bd=\bb-\bn-\bA \mbx. 
\eeq
Small gridding error requires  
that the objects are {\em a priori} close
 to the grid points. 
The gridding error is related to basis mismatch analyzed
in \cite{Chi2}.

\commentout{
$$ \begin{bmatrix}
     \ba_{1} &  \ldots & \ba_{s}  & \ba_{s+1} & \ldots & \ba_{M}
  \end{bmatrix}
  \begin{bmatrix}
     x_{1} \\
     \vdots \\
     x_{s} \\
     x_{s+1} \\
     \vdots \\
     x_{M}
  \end{bmatrix}
  + \mbe = \bb
$$
}

Sparse reconstruction with 
$N, s \ll M$,  where $s$ is the sparsity of   $\mbx$, has recently 
attracted a lot of attention in various areas thanks to
the breakthroughs in compressive sensing (CS)  \cite{CT, Don, Tib}. 
The main thrust of CS is  the
$L^1$-minimization principle,  Basis Pursuit (BP)  and Lasso,  
for solution characterization. Many $L^1$-based algorithms as well as
 the alternative,  greedy algorithms, which are
  not directly based on global optimization,  require either incoherence or Restricted Isometry Property (RIP)  to have good
  performances.

One commonly used characterization of  incoherence in CS is 
 in terms of the mutual
coherence  $\mu$.  
Let  the pairwise coherence between the $k$-th and $j$-th columns be 
\begin{equation}
  \mu (k,l) = {|\lan \ba_{k},\ba_{l}\ran |\over |\ba_k| |\ba_l|}. 
  \label{innerproduct}
\end{equation}
The mutual coherence of $\bA$ is the maximum pairwise
coherence among all pairs of columns 
\begin{equation}
   \mu(\bA) = \max_{j \neq l} \mu (k,l).
   \label{coherence}
\end{equation}

\commentout{
Let the condition number of a $n\times m$ ($n\geq m$) matrix be defined as
the ratio of its largest singular value to its $m$-th  singular value. 
It can be shown that the condition number of any 
submatrix of at most $k$ columns is at most
\[
\sqrt{1+(k-1)\mu(\bA) \over 1-(k-1)\mu(\bA)}
\]
provided that $(k-1)\mu(\bA)<1$. In other words, for any given $k$
the smaller the mutual coherence the better conditioned 
any submatrix of at most $k$ columns. 
}

According to theory of optimal recovery \cite{Don1}, 
for time sampling in $[0,1]$, the minimum resolvable length in
the frequency domain is unity. This is  the Rayleigh threshold
and we shall refer to this  length  as the Rayleigh length (RL). 
Hence for the traditional  inversion methods to work, it is
essential that the grid spacing in $\cG$ is no less than 1 RL.
In the CS setting  the Rayleigh threshold is closely
related to the decay property of the mutual coherence 
 \cite{crs}. Moreover,  for $\cG\subset \IZ$ 
and uniformly randomly selected $t_k\in [0,1]$ 
 the corresponding matrix  $\bA$ is 
a random partial Fourier matrix 
which  has a decaying mutual coherence  
 $\mu=\cO(N^{-1/2})$ and satisfies RIP with high probability  \cite{CT, Rau}.


\begin{figure}[t]
\includegraphics[width=8cm]{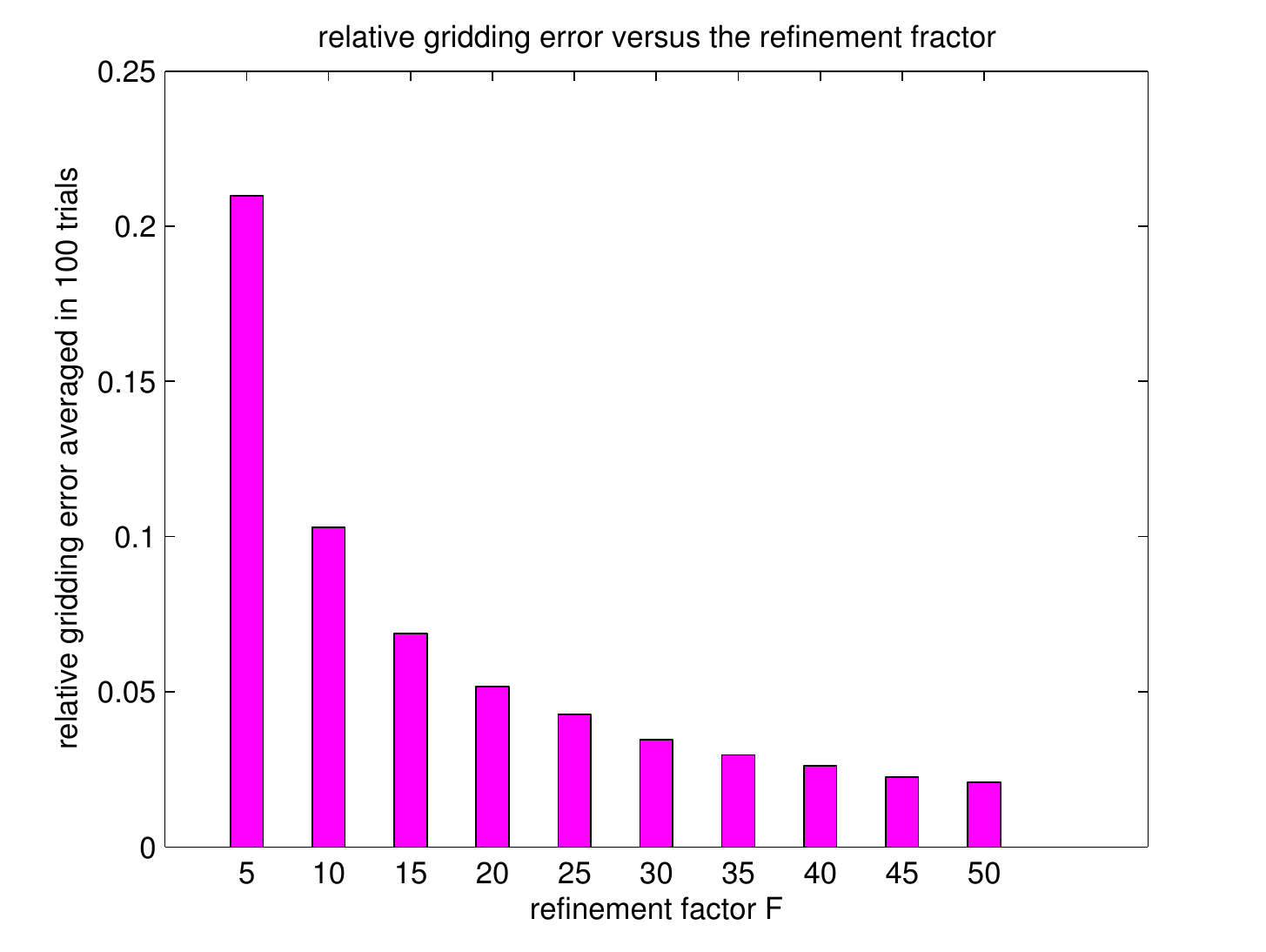}
\caption{
The relative gridding error is roughly inversely proportional to the refinement factor.  }
\label{fig328-1}
\end{figure}

Without any prior information about the object support, the gridding error for the resolved grid, however, can be  as large
as the data themselves, creating a unfavorable condition for sparse reconstruction.  To reduce the gridding error, 
it is natural to  consider 
the fractional  grid  
\beq
\label{23}
\IZ/F=\{ j/F: j\in \IZ\}
\eeq
 with
some large integer $F\in \IN$ which is {\em the
refinement factor}. Figure \ref{fig328-1}  shows that the
relative gridding error $\|\bd\|_2/\|\bb\|_2$ is roughly inversely proportional to the refinement factor.
The mutual coherence, however,  increases 
 with $F$ as  the near-by columns of the sensing matrix
become highly correlated.

Figure \ref{fig1} shows the coherence pattern $[\mu(j,k)]$ of
a $100\times 4000$ matrix (\ref{3}) 
with $F=20$ (left panel). The bright diagonal band represents a heightened correlation (pairwise coherence) between
a column vector and its neighbors  on both sides (about 30). 
The right panel of Figure \ref{fig1} shows  a half cross section of the coherence band across two RLs.     
Sparse recovery with large $F$ exceeds the capability of currently known
algorithms as the condition number of the  $100\times 30$ 
submatrix corresponding to the coherence band in Figure \ref{fig1}
easily exceeds  $10^{15}$. The high condition number makes 
stable recovery  impossible.  

\commentout{
\begin{figure}[t]
  \includegraphics[width=5cm]{CoFrac10M2000N100.eps}
   \includegraphics[width=5cm]{CoFrac10M2000N500.eps}
   \includegraphics[width=6cm]{April22/CoherenceBand.eps}
 \caption{Coherence pattern $[\mu(j,k)]$ for the $100\times 2000$ (left) and the $500\times 2000$ matrices with $F=10$ (middle). 
 The off-diagonal elements tend to diminish as the row number increases. The coherence band near the diagonals, however, persists, and has the profile shown on the right panel.}
  \label{figurecoherencepattern}\label{fig1}
\end{figure}
}
\begin{figure}[t]
  \includegraphics[width=8cm]{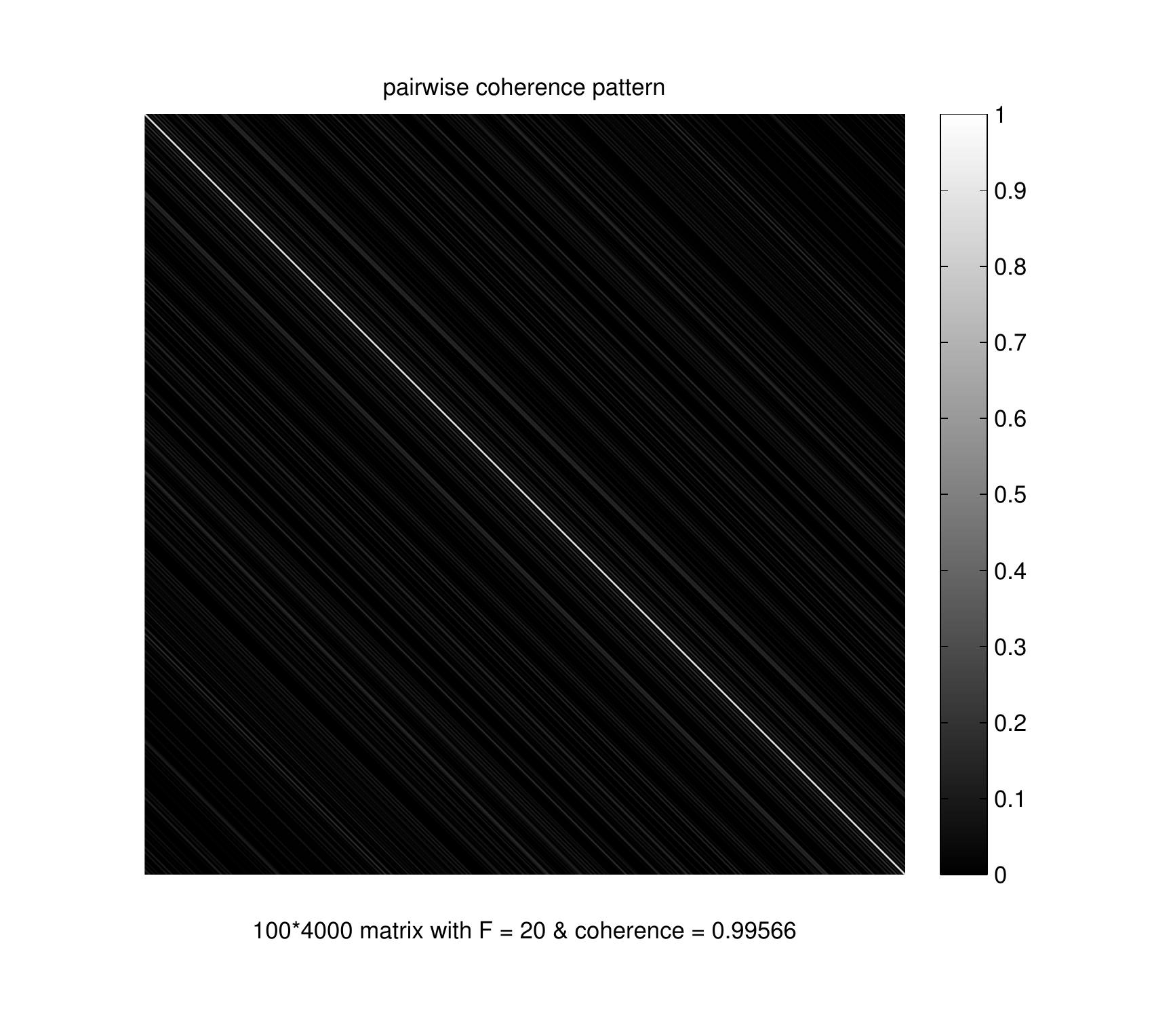}
   \includegraphics[width=8cm,height=7cm]{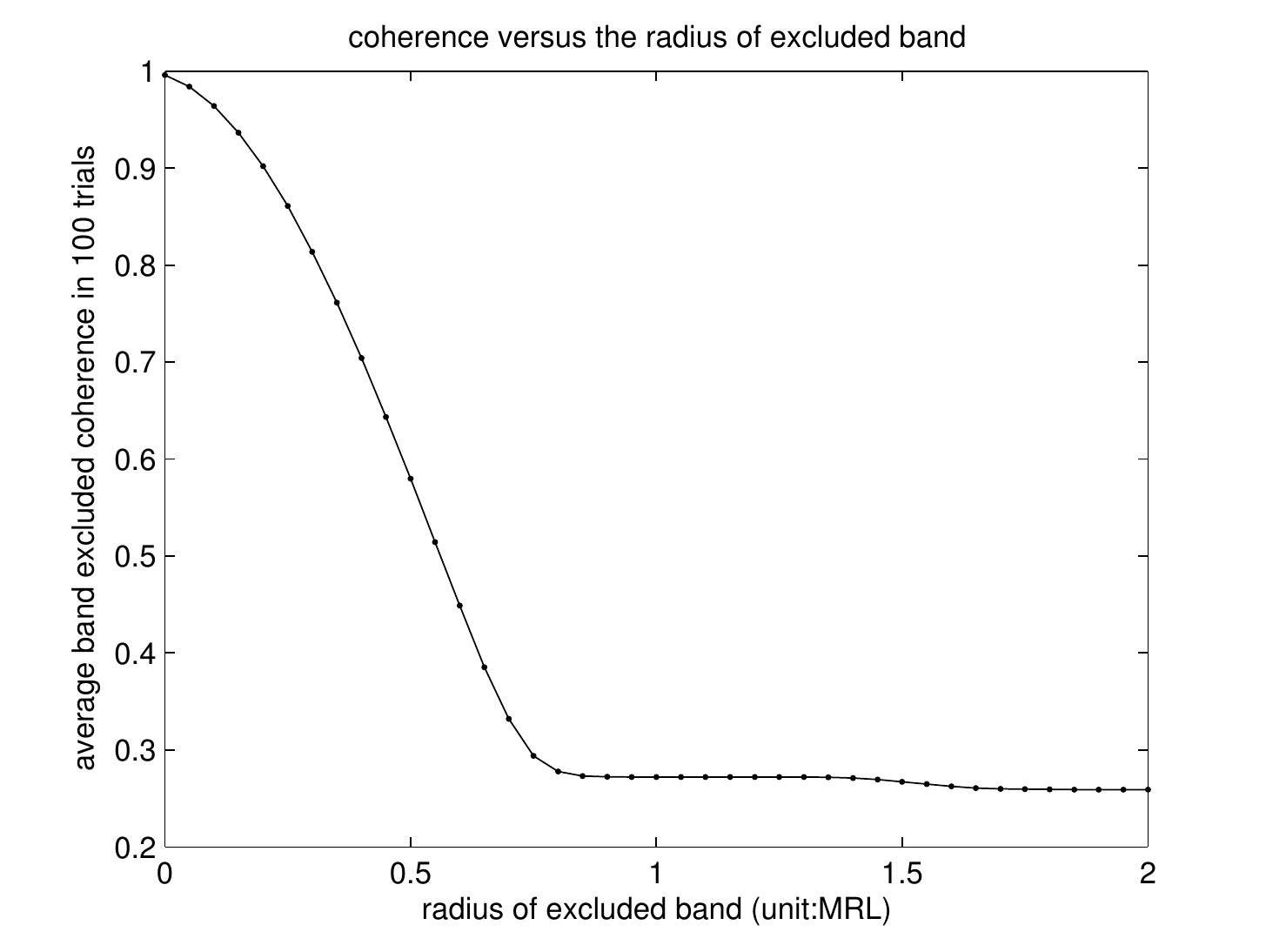}
 \caption{Coherence pattern $[\mu(j,k)]$ for the $100\times 4000$  matrix with $F=20$ (left).   The off-diagonal elements tend to diminish as the  row number increases. The coherence band near the diagonals, however, persists, and has the profile shown on the right panel where the vertical axis is the pairwise coherence and
 the horizontal axis is the separation between two columns in the unit of RL.}
  \label{figurecoherencepattern}\label{fig1}
\end{figure}

The difficulty with unresolved grids is not limited to
the problem of spectral estimation in signal processing.
Indeed, the issue  is intrinsic and  fundamental to  discretization of PDE-based inverse problems such as remote sensing and medical imaging \cite{CB, CK, Nat}. While  Figure \ref{fig1} is typical of
the coherence pattern from discretization of one-dimensional problem.  In two or three dimensions, the
coherent pattern is more complicated than Figure \ref{fig1}.
Nevertheless the coherence band typically 
reflects proximity in the physical space. 
The proximity between the  object support and its reconstruction
can be described by the Bottleneck  or the Hausdorff distance 
\cite{EI}. More generally, coherent bands can arise
in sparse and redundant  representation by  overcomplete dictionaries 
(see Section \ref{sec6} for an example). Under this circumstance, 
the Bottleneck or Hausdorff distance may not have a direct
physical meaning. 

In any case, the hope is that if the objects
are sufficiently  separated with respect to 
the coherence band,  then the problem
of  a huge  condition number associated with unresolved grids
can be somehow  circumvented and
the object support can be approximately reconstructed. 


Under this additional assumption of widely separated objects, we propose in  the present work
several algorithmic approaches to recovery with  unresolved grids
and provide some performance guarantee for these algorithms. 

The paper is organized as follows. In Section \ref{sec1} we introduce the technique of band exclusion (BE) to modify the Orthogonal Matching Pursuit (OMP)  and obtain 
a
performance guarantee for the improved  algorithm, called Band-excluded OMP (BOMP). In Section \ref{sec2} we introduce
the technique of Local Optimization (LO) and propose
the algorithms, Locally Optimized OMP (LOOMP) and
Band-excluded LOOMP (BLOOMP). In Section \ref{sec3} we introduce
the technique of Band-Excluded Thresholding (BET), which comes
in two form, Band-excluded Matched Thresholding (BMT)
and Band-excluded Locally Optimized Thresholding (BLOT) and propose
the algorithms, Band-excluded, Locally Optimized Subspace Pursuit (BLOSP), Band-excluded, Locally Optimized CoSaMP (BLOCoSaMP),
Band-excluded, Locally Optimized Iterative Hard Thresholding (BLOIHT) and BP/Lasso  with BLOT (BP/Lasso-BLOT). 
In Section \ref{sec5} we present numerical study of
the comparative advantages of various algorithms.
In Section \ref{sec6}  we compare the performance of our
algorithms with the existing  algorithms, Spectral Iterative Hard Thresholding (SIHT) and coherent-dictionary-based BP
recently proposed in \cite{DB} and \cite{CEN}, respectively.
We conclude in Section \ref{sec7}.  

\section{Band Exclusion  (BE)}\label{sec1}
The first technique that we introduce to take advantage of the
prior information of widely separated objects  is called Band Exclusion and 
can be easily embedded in  the greedy algorithm,
Orthogonal Matching Pursuit (OMP) \cite{DMA, PRK}. 

First let us recall a standard performance guarantee for  OMP  \cite{DET}.
\begin{proposition}
 Suppose that the sparsity $s$ of the signal vector $\mbx$ satisfies 
\begin{equation}
   \mu(\bA)(2s-1) + 2 \frac{\|\mbe\|_2}{\xmin} < 1 
   \label{donohocoherence} 
\end{equation}
where $\xmin =\displaystyle \min_{k} |x_k| = |x_s|$. Denote by $\hat{\mbx}$, the output of the OMP reconstruction. Then
\begin{equation*}
    \text{supp}(\hat{\mbx}) = \text{supp}(\mbx)
\end{equation*}
where $\text{supp}(\mbx)$ is the support of $\mbx$.\label{prop1}
\end{proposition}
 In the ideal  case where $\mbe=0$, (\ref{donohocoherence}) reduces
to 
\beq
\label{10.1}
\mu(\bA)<{1\over 2s -1} 
\eeq
which is near the threshold of OMP's capability for exact reconstruction of arbitrary objects of sparsity $s$. 

Intuitively speaking,   if the objects are {\em not}  in each other's coherence
band, then it should be possible to localize the
objects {\em approximately} within their respective
coherence bands, no matter how large the mutual coherence is.  

Let us first define precisely the notion of coherence band.
 Let $\eta>0$. Define the $\eta$-coherence band  of the index $k$ to be the set
\begin{equation}
   B_\eta(k) = \{i\ | \ \mu(i,k) > \eta\},
   \label{singleneighbor}
\end{equation}
and the $\eta$-coherence band of the index set $S$  to be the set
\begin{equation*}
   B_\eta(S) = \cup_{k \in S} B_\eta(k).
\end{equation*}
Due to the symmetry $\mu(i,k)=\mu(k,i),\forall i,k$,
$i\in B_\eta(k)$ if and only if $k\in B_\eta(i)$. 

Denote
\beq
   B^{(2)}_\eta(k) &\equiv& B_\eta(B_\eta(k))= \displaystyle \cup_{j\in B_\eta(k)} B_\eta(j)
   \label{doubleneighbor}\\
   B^{(2)}_\eta(S) &\equiv& B_\eta(B_\eta(S))= \cup_{k \in S} B^{(2)}_\eta(k).
\eeq

To imbed BE into OMP, 
we make the following change to  the matching step 
\[
i_{\rm max} = \hbox{arg}\min_{i}|\lan \br^{n-1},\ba_i\ran | , \quad i \notin B^{(2)}_\eta(S^{n-1}),\quad n=1,2,....
\]
meaning that  the double $\eta$-band of the estimated  support  in the previous iteration is avoided in the current search. This is natural if the sparsity pattern
of the object is such that
$B_\eta(j), j\in \hbox{supp}(\mbx)$ are pairwise disjoint. 
We call the modified algorithm  the Band-excluded Orthogonal Matching Pursuit (BOMP) which is formally stated in  {\bf Algorithm 1}.

\bigskip

\begin{center}
   \begin{tabular}{l}
   \hline
   
   \centerline{{\bf Algorithm 1.}\quad Band-Excluded Orthogonal Matching Pursuit (BOMP)} \\ \hline
   Input: $\bA, \bb,\eta>0$\\
 Initialization:  $\mbx^0 = 0, \br^0 = \bb$ and $S^0=\emptyset$ \\ 
Iteration: For  $n=1,...,s$\\
\quad {1) $i_{\rm max} = \hbox{arg}\min_{i}|\lan \br^{n-1},\ba_i\ran | , i \notin B^{(2)}_\eta(S^{n-1}) $} \\
  \quad      2) $S^{n} = S^{n-1} \cup \{i_{\rm max}\}$ \\
  \quad  3) $\mbx^n = \hbox{arg} \min_\bz \|
     \bA \bz-\bb\|_2$ s.t. \hbox{supp}($\bz$) $\in S^n$ \\
  \quad   4) $\br^n = \bb- \bA \mbx^n$\\
 Output: $\mbx^s$. \\
 \hline
   \end{tabular}
\end{center}

\bigskip

 A main theoretical  result of the present paper is
  the following  performance guarantee for BOMP. 
\begin{theorem}
\label{thm1}

Let $\mbx$ be $s$-sparse. Let $\eta>0$ be fixed. 
Suppose that
\beq
\label{sep}
B_\eta(i)\cap B^{(2)}_\eta(j)=\emptyset, \quad \forall i, j\in \hbox{supp}(\mbx)
\eeq
and that
    \beq
    \eta(5s -4)\frac{\xmax}{\xmin} + \frac{5\|\mbe\|_2}{2\xmin} < 1 
    \label{RMIP}
    \eeq
    where
    \[
    \xmax = \max_{k} |x_k|,\quad  \xmin = \min_{k} |x_k|.
    \]
    Let $\hat\mbx$ be the BOMP reconstruction. 
Then $\hbox{supp}(\hat\mbx)\subseteq B_\eta(\hbox{supp}(\mbx))$
and moreover every nonzero component of $\hat\mbx$ is in
the $\eta$-coherence band of a unique nonzero component of $\mbx$. 
  \label{momp}
\end{theorem}

\begin{proof} 
We  prove the theorem  by induction. 

Suppose  $\hbox{supp} (\mbx)=\{J_1,\ldots, J_s\}$. 
Let $J_{\rm max}\in \hbox{supp}(\mbx)$ be the index of
the largest component in absolute value of $\mbx$. 

In the first step,
\beq
 \label{14}  |\bb^{\star}\ba_{J_{\rm max}}| & =& |{x}_{J_{\rm max}} + \bx_{J_2}\ba_{J_2}^{\star}\ba_{J_{\rm max}} + ... + \bx_{J_s}\ba_{J_s}^{\star}\ba_{J_{\rm max}} + \mbe^{\star}\ba_{J_{\rm max}}| \\
   & \ge& x_{\rm max} - x_{\rm max}(s-1)\eta - \|\mbe\|_2\nn
\eeq
by assumption (\ref{sep}). 
On the other hand, $\forall l \notin B_\eta(\hbox{supp}(\mbx))$,
\beq
\label{15}  |\bb^{\star}\ba_l| & =& |\bx_{J_1}\ba_{J_1}^{\star}\ba_l + \bx_{J_2}\ba_{J_2}^{\star}\ba_l + ... + \bx_{J_s}\ba_{J_s}^{\star}\ba_l + \mbe^{\star}\ba_l| \\
  & \le& x_{\rm max} s\eta + \|\mbe\|_2\nn
\eeq
by using (\ref{sep}) again.

Hence,  if
 \[
 (2s-1)\eta + 2 \frac{\|\mbe\|_2}{x_{\rm max}} < 1,
 \]
 then the right hand side of (\ref{14}) is greater than
 the right hand side of (\ref{15}) which implies that
 the first  index selected by BOMP must belong to $B_\eta(\hbox{supp}(\mbx))$. 

Now suppose without loss of generality that the first $(k-1)$ indices $I_1,...,I_{k-1}$ selected by BOMP  are in $B_\eta(J_i), J_i \in \hbox{supp}(\mbx), i=1,...,k-1$,
respectively.  
Write the residual as 
\begin{align*}
  \br^{k-1} & = \bb - c_{I_1}\ba_{I_1} - c_{I_2}\ba_{I_2} - ... - c_{I_{k-1}}\ba_{I_{k-1}}.
\end{align*}
First, we estimate the coefficients $c_{I_1},...,c_{I_{k-1}}$. Since $\lan\br^{k-1},\ba_{I_1}\ran  = 0$,
\begin{equation*}
  {c}_{I_1} = \bx_{J_1}\ba_{J_1}^{\star}\ba_{I_1} + \bx_{J_2}\ba_{J_2}^{\star}\ba_{I_1} + ... + \bx_{J_s}\ba_{J_s}^{\star}\ba_{I_1} + \mbe^{\star}\ba_{I_1} - c_{I_2}\ba_{I_2}^{\star}\ba_{I_1} - ... - c_{I_{k-1}}\ba_{I_{k-1}}^{\star}\ba_{I_1},
\end{equation*}
which implies
\beqn
  |c_{I_1}| \le \xmax + \xmax (s-1)\eta + \|\mbe\|_2 + \eta(|c_{I_2}|+|c_{I_3}|+...+|c_{I_{k-1}}|)
\eeqn
Likewise, we have
\beq
  |c_{I_j}| \le \xmax + \xmax (s-1)\eta + \|\mbe\|_2 + \eta
  \sum_{i\neq j} |c_{I_i}|,\quad j=1,...,k-1. \label{17}
\eeq

Let $\emax = \displaystyle \max_{j=1,...,k-1} |c_{I_{j}}|$.  Inequality 
(\ref{17}) implies  that
\beqn
  \emax \le \xmax + \xmax (s-1)\eta + \|\mbe\|_2 + \eta(k-2)\emax 
  \eeqn
and hence
\beqn
  \emax \le \frac{1}{1-\eta(k-2)}[\xmax + \xmax (s-1)\eta + \|\mbe\|_2]
\eeqn
Moreover, condition (\ref{RMIP}) implies that $\eta(s-1)<\frac{1}{5}$ and $\frac{1}{1-\eta(k-2)} \le \frac{5}{4}$. Hence
\beq
\label{20}
\emax \le \frac{5}{4}(\xmax +\frac{1}{5} \xmax  + \|\mbe\|_2) \le \frac{3}{2} \xmax + \frac{5}{4}\|\mbe\|_2.
\eeq

We claim  that $B^{(2)}_\eta(S^{k-1})$ and $\{J_k,...,J_s\}$ are
disjoint. 

If the claim is not true, then there exists $J_i,$ for
some $i\in \{k,\ldots, s\}$,  $J_i\in B^{(2)}_\eta(I_l)$ for some
$l\in \{1,...,k-1\}$. Consequently, $J_i\in B^{(3)}_\eta(J_l)$ or equivalently $B_\eta(J_i)\cap B^{(2)}_\eta(J_l)\neq \emptyset$ 
 which
is contradictory to the assumption (\ref{sep}).

Now we show that the index selected  in the $k$-th step is in $B_\eta (\{J_k,...,J_s\})$. 

On the one hand, we have
\beq
  \label{21}|\br^{k-1 {\star}}\ba_{J_i}| &=& |\bx_{J_1}\ba_{J_1}^{\star}\ba_{J_i}+\ldots + \bx_{J_s}\ba_{J_s}^{\star}\ba_{J_i} + \mbe^{\star}\ba_{J_i}\\
  &&\quad\nn -c_{I_1}\ba_{I_1}^{\star}\ba_{J_i}-...-c_{I_{k-1}}\ba_{I_{k-1}}^{\star}\ba_{J_i}| \\
 &\ge &x_{\rm min} - \eta(s-1)\xmax - \|\mbe\|_2 - \eta(k-1)\emax. \nn
\eeq
One the other hand,  we have that $\forall l \notin B^{(2)}_\eta(S^{k-1}) \cup B_\eta(\{J_k,...,J_s\})$,
\beq
\label{22}  |\br^{k-1 {\star}}\ba_{l}|  &=& |\bx_{J_1}\ba_{J_1}^{\star}\ba_{l}+...+\bx_{J_k}\ba_{J_k}^{\star}\ba_{l} + ... + \bx_{J_s}\ba_{J_s}^{\star}\ba_{l} + \mbe^{\star}\ba_{l}\\
&&\quad  - c_{I_1}\ba_{I_1}^{\star}\ba_{l}-...-c_{I_{k-1}}\ba_{I_{k-1}}^{\star}\ba_{l}| \nn\\
  & \le& \eta s \xmax + \|\mbe\|_2 + \eta (k-1) \emax\nn 
\eeq
in view of  $B_\eta(\{J_1,\ldots,J_{k-1}\})
\subseteq B^{(2)}_\eta(S^{k-1})$ as a result of
the induction assumption.

If the right hand side of (\ref{21}) is greater than the right hand
side of (\ref{22}) or equivalently
\beq
\label{31}
\eta(2s + 3k -4)\frac{\xmax}{\xmin} + 
(2+{5\over 2} \eta (k-1))\frac{\|\mbe\|_2}{\xmin} < 1
\eeq
then the $k$-th index selected by BOMP  must be in $B_\eta(\{J_k,...,J_s\})$ because
\[
B^{(2)}_\eta(S^{k-1}) \cap \{J_k,...,J_s\}=\emptyset
\]
and because the $k$-th selected index does not belong in $B^{(2)}_\eta(S^{k-1})$ according to  the band-exclusion rule. 
 Condition (\ref{RMIP}) implies (\ref{31}) by 
setting  the maximal $k=s$ in (\ref{31}) and noting
that $\eta(5s -4)<1$ under (\ref{RMIP}). 
\end{proof}

\begin{remark}
\label{rmk1}

In the case of the matrix 
(\ref{3}),  if every two indices in $\suppx$ is more than
one RL apart, then $\eta$ is small for sufficiently
large $N$, cf. Figure \ref{fig1}. 
  
 When the dynamic range ${\xmax}/{\xmin} = \cO(1)$, 
 Theorem \ref{thm1} guarantees approximate recovery
 of  $\cO(\eta^{-1})$ sparsity pattern by BOMP.

\end{remark}

\begin{remark}\label{rmk2}
The main difference between Theorem $\ref{thm1}$ and 
Proposition \ref{prop1} lies in the role
played by  the dynamic range $\xmax/\xmin$
and condition (\ref{sep}).

First, numerical evidence points  to degradation in BOMP's performance
for large dynamic ranges (Figure \ref{figB}). This is consistent with
the prediction of (\ref{RMIP}). 

Secondly, condition  (\ref{sep}) means 
that BOMP can resolve 3 RLs.
Numerical experiments show that 
BOMP can resolve 
 objects separated by close to 1  RL  when the dynamic
 range is close to 1 (Figure \ref{fig422}).

\end{remark}

\section{Local Optimization (LO)}\label{sec2}

As our numerical experiments show, the main shortcoming  with
BOMP is in its failure to perform even when the dynamic range is  only moderate. 

To overcome this problem, we now introduce 
the second technique: 
the {\em Local Optimization} (LO).  

LO is a residual-reduction technique  applied
to the current estimate $S^k$ of the object support.  
To this end,  we minimize  the residual  ${\|\bA \hat\mbx-\bb\|_2}$ by varying  one location at a time 
while all other locations held fixed.
In each step we consider  $\hat\mbx$ whose support
differs from $S^n$ by at most one index  in the  coherence band of $S^n$ but whose amplitude is chosen to minimize
the residual. The search is local in the sense that
  during the search in the coherence band of one nonzero component
the locations of other nonzero components are fixed.
The amplitudes of the improved estimate  is carried out by solving the least squares problem. Because of
the local nature of the LO step, the computation is not
expensive. 

\bigskip

\begin{center}
   \begin{tabular}{l}
   \hline  
   \centerline{{\bf Algorithm 2.}\quad  Local Optimization (LO)}  \\ \hline
    Input:$\bA,\bb, \eta>0,  S^0=\{i_1,\ldots,i_k\}$.\\
Iteration:  For $n=1,2,...,k$.\\
\quad 1) $\mbx^n= \hbox{arg}\,\,\min_{\bz}\|\bA \bz-\bb\|_2, \hbox{supp}(\bz)=(S^{n-1} \backslash \{i_n\})\cup \{j_n\}, $ for some $  j_n\in B_\eta(\{i_n\})$.\\
 \quad 2) $S^n=\hbox{supp}(\mbx^n)$.\\
    Output:  $S^k$.\\
    \hline
   \end{tabular}
\end{center}

\bigskip

Embedding LO in BOMP gives rise to the Band-excluded, Locally
Optimized Orthogonal Matching Pursuit (BLOOMP).

\bigskip

\begin{center}
   \begin{tabular}{l}
   \hline
   
   \centerline{{\bf Algorithm 3.} Band-excluded, Locally Optimized Orthogonal Matching Pursuit (BLOOMP)} \\ \hline
   Input: $\bA, \bb,\eta>0$\\
 Initialization:  $\mbx^0 = 0, \br^0 = \bb$ and $S^0=\emptyset$ \\ 
Iteration: For  $n=1,...,s$\\
\quad {1)  $i_{\rm max} = \hbox{arg}\min_{i}|\lan \br^{n-1},\ba_i\ran | , i \notin B^{(2)}_\eta(S^{n-1}) $} \\
  \quad      2) $S^{n} = \hbox{LO}(S^{n-1} \cup \{i_{\rm max}\})$ where $\hbox{LO}$ is the output of Algorithm 2.\\
  \quad  3) $\mbx^n = \hbox{arg}  \min_\bz \|
     \bA \bz-\bb\|_2$ s.t. \hbox{supp}($\bz$) $\in S^n$ \\
  \quad   4) $\br^n = \bb- \bA \mbx^n$\\
 Output: $\mbx^s$. \\
 \hline
   \end{tabular}
\end{center}

\bigskip

We now give a condition under which LO does not spoil the
BOMP reconstruction of Theorem \ref{thm1}. 
\begin{theorem}
\label{thm2}
Let $\eta>0$ and let $\mbx$ be a $s$-sparse vector such
that (\ref{sep}) holds. 
Let $S^0$ and $S^k$ be the input and output, respectively,  of the LO algorithm. 

If
\beq
x_{\rm min}> (\ep+2(s-1)\eta) \lt({1\over 1-\eta}+\sqrt{{1\over (1-\eta)^2}+{1\over 1-\eta^2}}\rt)\label{93}
\eeq
and each element of $S^0$  is in the $\eta$-coherence
band of a unique nonzero component of $\mbx$, then
 each element of $S^k$ remains in the $\eta$-coherence
band of a unique nonzero component of $\mbx$. 

\end{theorem}
\begin{proof}
Because the iterative nature of Algorithm 2, it is sufficient to show that
each element of $S^1$ is in the $\eta$-coherence
band of a unique nonzero component of $\mbx$. 

Suppose $J_1\in \supp{(\mbx)}$ and $i_1\in B_\eta (J_1)$.  
Let 
\beqn
r&=&\min_{\bz}\|\bA \bz-\bb\|_2, \quad \hbox{supp}(\bz)=(S^{0} \backslash \{i_1\})\cup \{J_1\} \\
r'&=&\min_{\bz}\|\bA \bz-\bb\|_2, \quad \hbox{supp}(\bz)=(S^{0} \backslash \{i_1\})\cup \{j\},\quad  j\in B_\eta(i_1)\backslash  B_\eta(J_1).
\eeqn
We want to show that $r<r', \forall j\in B_\eta(i_1)\backslash  B_\eta(J_1)$ so that the LO step is certain to pick a new index 
within the $\eta$-coherence band of $J_1$, reducing the residual
in the meantime.
For the subsequent analysis, we fix $ j\in B_\eta(i_1)\backslash  B_\eta(J_1)$. 

Reset the $J_1$ component of $\mbx$ to zero
and denote the resulting vector by $\mbx'$. Hence the sparsity of $\mbx'$ is $s-1$.  It follows from the definition of $r$ that
\[
r\leq \min_{\bz}\|\bA \bz-\bA \mbx'-\mbe\|_2, \quad \hbox{supp}(\bz)=\{i_2,\ldots,i_k\}.
\]
We also have
\beqn
r'&=&\min_{\bz} \|\bA(\bz-\mbx')-\mbe+ x_{J_1} \ba_{J_1}-c \ba_j\|_2,\quad \supp(\bz)=\{i_2,\ldots,i_k\},\quad c\in \IC
\eeqn
and hence by the law of cosine
\beq
\label{91}
r' &\geq&\min_{\bz, c}\sqrt{\|\bA(\bz-\mbx')-\mbe\|^2_2+ \|x_{J_1} \ba_{J_1}-c \ba_j\|^2_2-2|\lan \bA(\bz-\mbx')-\mbe, x_{J_1} \ba_{J_1}-c \ba_j\ran | }
\eeq
where $\supp(\bz)=\{i_2,\ldots,i_k\}$.

Because of (\ref{sep}),  $j, J_1\not\in B_\eta( \supp(\mbx')\cup \{i_2,\ldots,i_k\})$. By the definition of $\eta$-coherence band,
we have 
\[
|\lan \bA(\bz-\mbx')-\mbe, x_{J_1} \ba_{J_1}-c \ba_j\ran |  \leq
(\ep+\eta(s+k-2))(|x_{J_1}|+|c|)
\]
and hence by (\ref{91})
\beqn
\label{92}
r' &\geq&\min_{\bz, c}\sqrt{\|\bA(\bz-\mbx')-\mbe\|^2_2+ |x_{J_1}|^2 +|c|^2-2\eta|cx_{J_1}|
-2(\ep+\eta(s+k-2))(|x_{J_1}|+|c|)}\\
&\geq & \sqrt{\min_{\bz}\|\bA(\bz-\mbx')-\mbe\|_2^2
+\min_{c\in \IC}\lt[ |x_{J_1}|^2 +|c|^2-2\eta|cx_{J_1}|
-2(\ep+\eta(s+k-2))(|x_{J_1}|+|c|)\rt]}. \nn
\eeqn
To prove $r<r'$, it suffices to show
\beqn
&&\min_{c\in \IC}\lt[ |x_{J_1}|^2 +|c|^2-2\eta|cx_{J_1}|
-2(\ep+\eta(s+k-2))(|x_{J_1}|+|c|)\rt]\\
&=&(1-\eta^2)|x_{J_1}|^2-2(1+\eta)(\ep+\eta(s+k-2))|x_{J_1}|-
(\ep+\eta(s+k-2))^2>0
\eeqn
which leads to the inquality
\[
|x_{J_1}|>  (\ep+(s+k-2)\eta) \lt({1\over 1-\eta}+\sqrt{{1\over (1-\eta)^2}+{1\over 1-\eta^2}}\rt).
\]
Considering the worst case scenario, we replace $|x_{J_1}|$ by $x_{\rm min}$ and
$k$ by $s$ to  obtain the condition (\ref{93}). 

\end{proof}

\begin{corollary}
Let $\hat \mbx$ be the output of BLOOPM. 
Under the assumptions of Theorems \ref{thm1} and \ref{thm2},
 $\hbox{supp}(\hat\mbx)\subseteq B_\eta(\hbox{supp}(\mbx))$
and moreover every nonzero component of $\hat\mbx$ is in
the $\eta$-coherence band of a unique nonzero component of $\mbx$. 
\end{corollary}
Even though we can not 
improve the performance guarantee for BLOOMP, in practice the LO technique 
greatly enhances the success probability of recovery that
BLOOMP has
the best performance  among all the algorithms tested with respect to noise stability and dynamic range  (see Section \ref{sec5}).   In particular,  the LO step greatly enhances
the performance of BOMP  w.r.t. dynamic range.

\commentout{ 
 For the purpose of comparison, we shall also consider
\begin{center}
   \begin{tabular}{l}
   \hline
   \centerline{{\bf Algorithm 3.} Locally Optimized Orthogonal Matching Pursuit (LOOMP)} \\ \hline
   Input: $\bA, \bb$\\
 Initialization:  $\mbx^0 = 0, \br^0 = \bb$ and $S^0=\emptyset$ \\ 
Iteration: For  $n=1,...,s$\\
\quad {1)  $i_{\rm max} = \hbox{arg}\min_{i}|\lan \br^{n-1},\ba_i\ran | , i \notin S^{n-1} $} \\
  \quad      2) $S^{n} = \hbox{LO}(S^{n-1} \cup \{i_{\rm max}\})$ where $\hbox{LO}$ is the output of Algorithm 2.\\
  \quad  3) $\mbx^n = \hbox{arg} \min_\bz \|
     \bA \bz-\bb\|_2$ s.t. \hbox{supp}($\bz$) $\in S^n$ \\
  \quad   4) $\br^n = \bb- \bA \mbx^n$\\
 Output: $\mbx^s$. \\
 \hline
   \end{tabular}
\end{center}
}


\section{Band-Excluded  Thresholding (BET)}\label{sec3}

The BE technique can be extended and applied
to selecting $s$ objects all at once in what is
called the Band-Excluded Thresholding (BET).

We consider two forms of BET. The first is the Band-excluded Matched  Thresholding (BMT) which  is the band-exclusion version of
the One-Step Thresholding (OST) recently shown to possess compressed-sensing
capability under incoherence conditions  \cite{BCJ}.

For the purpose of comparison with BOMP,  we give a performance 
guarantee for BMT under  similar but weaker conditions than (\ref{sep})-(\ref{RMIP}).  

\bigskip
\begin{center}
   \begin{tabular}{l}\hline
    \centerline{{\bf Algorithm 4.}\quad Band-excluded Matched Thresholding (BMT) } \\ \hline
    Input: $\bA, \bb,\eta>0$. \\
    Initialization: $S^0=\emptyset$.\\
   Iteration:  For $k=1,...,s$,\\
 \quad  1) $i_k=\hbox{arg}\max_j |\lan \bb,\ba_{j} \ran |, \forall j\notin B^{(2)}_\eta(S^{k-1})$. \\
 \quad 2) $S^k=S^{k-1}\cup\{i_k\}$\\
  Output $\hat \mbx=\hbox{arg}\min_{\bz}\|\bA \bz-\bb\|_2$ s.t.
    $\hbox{supp}(\bz)\subseteq S^s$\\
    \hline
   \end{tabular}
\end{center}

\begin{theorem}
 Let $\mbx$ be $s$-sparse. Let $\eta>0$ be fixed. 
Suppose that
\beq
\label{sep2}
B_\eta(i)\cap B_\eta(j)=\emptyset, \quad \forall i, j\in \hbox{supp}(\mbx)
\eeq
and that
    \beq
    \eta(2s -1)\frac{\xmax}{\xmin} + \frac{2\|\mbe\|_2}{\xmin} < 1
    \label{RMIP2}
    \eeq
    where
    \[
    \xmax = \max_{k} |x_k|,\quad  \xmin = \min_{k} |x_k|.
    \]
    Let $\hat\mbx$ be the BMT reconstruction. 
Then $\hbox{supp}(\hat\mbx)\subseteq B_\eta(\hbox{supp}(\mbx))$
and moreover every nonzero component of $\hat\mbx$ is in
the $\eta$-coherence band of a unique nonzero component of $\mbx$. 
\label{thm3}
\end{theorem}

\begin{proof}Let  $\hbox{supp} (\mbx)=\{J_1,\ldots, J_s\}$. 
Let $J_{\rm max}\in \hbox{supp}(\mbx)$ be the index of
the largest component of $\mbx$ in absolute value. 

On the one hand, for $k=1,...,s$,
\beq
 \label{25}   |\bb^{\star}\ba_k| & = &|\bx_1\ba_1^{\star}\ba_k + ... + \bx_{k-1}\ba_{k-1}^{\star}\ba_k + {x}_k +
    \bx_{k+1}\ba_{k+1}^{\star}\ba_k + ... + \bx_s\ba_s^{\star}\ba_k + \mbe^{\star}\ba_k| \\
    & \ge& x_{\rm min} - (s-1)\eta \xmax - \|\mbe\|_2.\nn
\eeq

On the other hand, $\forall l \notin B_\eta(\suppx)$,
\beq
\label{26}
  |\bb^{\star}\ba_l| & =& |\bx_1\ba_1^{\star}\ba_l + \bx_2\ba_2^{\star}\ba_l + ... + \bx_s\ba_s^{\star}\ba_l + \mbe^{\star}\ba_l| \\
  & \le& \xmax s\eta + \|\mbe\|_2.\nn
\eeq
Therefore, the condition (\ref{RMIP2}) 
implies that  the right hand side of (\ref{25}) is greater than
the right hand side of (\ref{26}). This means $|\bb^*\ba_k|> |\bb^*\ba_l|, \forall k=1,..,s, \forall l\not\in B_\eta(\suppx)$ and hence 
the $s$ highest points of $|\bb^*\ba_k|$ are in
$B_\eta(\suppx)$. 

From  step ii) of BMT and (\ref{sep2}) it follows the second half of the statement, namely 
every nonzero component of $\hat\mbx$ is in
the $\eta$-coherence band of a unique nonzero component of $\mbx$.
\end{proof}



 
 Condition (\ref{sep2}) roughly means that the objects
are separated by at {\bf two}  RLs which
is weaker than (\ref{sep}). 
In numerical simulations, however, BOMP performs far better than BMT.
In other words,  BMT is not a stand-alone algorithm
but should instead be imbedded in other algorithms such
as Subspace Pursuit (SP)  \cite{DM},  the Compressive Sampling Matching Pursuit (CoSaMP)  \cite{NT}  and  the Normalized Iterative Hard Thresholding (IHT)
\cite{BD2}. This gives rise to Band-excluded Subspace Pursuit (BSP), Band-excluded Compressive Sampling Matching Pursuit (BCoSaMP) and Band-excluded
Normalized Iterative Hard Thresholding (BNIHT)
which we demonstrate their performance numerically.

In addition to BMT,  the second form of BET, namely  the  Band-excluded, Locally Optimized Thresholding (BLOT),
can further enhance  the performance in reconstruction with
unresolved grids. 

\bigskip

\begin{center}
   \begin{tabular}{l}
   \hline
   
   \centerline{{\bf Algorithm 5.}\quad Band-excluded, Locally Optimized  Thresholding (BLOT)}  \\ \hline
    Input: $\mbx=(x_1,\ldots, x_M)$, $\bA, \bb, \eta>0$.\\
    Initialization:  $S^0=\emptyset$.\\
Iteration:  For $n=1,2,...,s$.\\
\quad 1) $i_n= \hbox{arg}\,\,\min_j |x_j|, j\not\in B^{(2)}_\eta(S^{n-1}) $.\\
 \quad 2)  $S^n=S^{n-1}\cup\{i_n\}$.\\
    Output: $\hat\mbx=\hbox{arg}\min \|\bA\bz-\bb\|_2$, $\hbox{supp}(\bz)\subseteq \hbox{LO}(S^s)$ where $\hbox{LO}$ is the output of Algorithm 2.\\
     \hline
   \end{tabular}
\end{center}

\bigskip

Now we state the algorithm Band-excluded, Locally Optimized
Subspace Pursuit (BLOSP). The Band-excluded Locally Optimized CoSaMP (BLOCoSaMP)  is  similar and omitted here. 

\bigskip

\begin{center}
   \begin{tabular}{l}
   \hline
   
      \centerline{{\bf Algorithm 6.} \quad Band-excluded, Locally Optimized Subspace Pursuit  (BLOSP)}  \\ \hline
   Input: $\bA,\bb, \eta>0$.\\
    Initialization: $\mbx^0=0,  \br^0 = \bb$\\
Iteration: For $n=1,2,...$,\\
\quad 1)  $\tilde{S}^{n} = \hbox{supp}(\mbx^{n-1}) \cup \hbox{supp} (\hbox{BMT}(\br^{n-1}))$\\
\quad \quad  where 
 $\hbox{ BMT}(\br^{n-1})$ is the output of
    Algorithm 4 with data $\br^{n-1}$.\\
\quad 2) $\tilde{\mbx}^n = \text{arg} \min \|\bA\bz - \bb\|_2$ s.t. $\hbox{supp}(\bz)\subseteq \tilde S^n $. \\
\quad 3) 
$S^n =\hbox{supp}( \hbox{BLOT} (\tilde{\mbx}^n))$ where $\hbox{BLOT}(\tilde{\mbx}^n)$ is the output of Algorithm 5.\\
\quad 4) $\br^n = \min_{\bz}\|\bA \bz-\bb\|_2, \hbox{supp}(\bz)\subseteq S^n $.\\
 \quad 5)  If $\|\br^{n-1}\|_2 \le \epsilon$ or $\|\br^n\|_2 \ge \|\br^{n-1}\|_2$, then quit and set $S=S^{n-1}$; otherwise continue iteration.\\
    Output: $\hat \mbx=\hbox{arg}\min_{\bz}\|\bA \bz-\bb\|_2$ s.t.
    $\hbox{supp}(\bz)\subseteq S$.\\
    \hline
   \end{tabular}
\end{center}

\bigskip

\commentout{

Our  numerical experiments show that given the same number of measurements BLOSP and BCoSaMP have
a higher probability of success than BLOOMP. 
A simple way of taking advantage of both algorithms is
to use the output of BLOOMP as the initial guess
of BLOSP and BCoSaMP. 
}
 
\bigskip

Embedding BLOT in NIHT turns out to have a nearly identical 
performance to embedding BLOT in the Iterative
Hard Thresholding (IHT) \cite{BD}. Since the latter
is simpler to implement and more efficient to compute,
we state the resulting algorithm,  the Band-excluded, Locally Optimized IHT (BLOIHT), below.  

\bigskip

\begin{center}
   \begin{tabular}{l}
   \hline
   
   \centerline{{\bf Algorithm 7.}\quad Band-excluded, Locally Optimized Iterative Hard Thresholding (BLOIHT)}  \\ \hline
    Input: $\bA, \bb, \eta>0$. \\
    Initialization:  $\hat \mbx^0=0,\br^0=\bb$.\\
Iteration:  For $n=1,2,...$,  \\
\quad 1) $\mbx^n= \hbox{BLOT} (\mbx^{n-1}+\bA^*\br^{n-1})$ where $\hbox{BLOT}$ denotes the output of Algorithm 5. \\
 \quad 2)  If $\|\br^{n-1}\|_2 \le \epsilon$ or $\|\br^n\|_2 \ge \|\br^{n-1}\|_2$, then quit and set $S=S^{n-1}$; otherwise continue iteration.\\
    Output:  $\hat\mbx$.\\
    \hline
   \end{tabular}
\end{center}

\bigskip


In addition, the technique BLOT can be used to enhance the recovery capability with unresolved grids  of 
 the $L^1$-minimization principles, Basis Pursuit (BP)
\beq
\min_{\bz}  \|\bz\|_1,\quad\hbox{subject to}\quad 
\bb=\bA \bz.
\eeq
 and the Lasso
 \beq
\min_{\bz} {1\over 2} \|\bb-\bA \bz\|_2^2+\lambda \sigma \|\bz\|_1,
\eeq
where $\sigma$ is the standard deviation of
the each noise component and $\lambda$ is the regularization parameter.  
In this case, BLOT  is applied to  the
BP and Lasso estimates $\hat \mbx$ to produce 
a $s$-sparse reconstruction. The resulting algorithms 
are  called  BP-BLOT and 
Lasso-BLOT, respectively.

The thresholded Lasso, the Lasso followed by a hard  thresholding,  has been considered previously (see \cite{MY} and references therein). The novelty of our version  lies in the
 BE and LO steps which greatly
enhance the performance in dealing with unresolved grids.  

\section{Numerical study}\label{sec5}

\begin{figure}
\includegraphics[width=8cm]{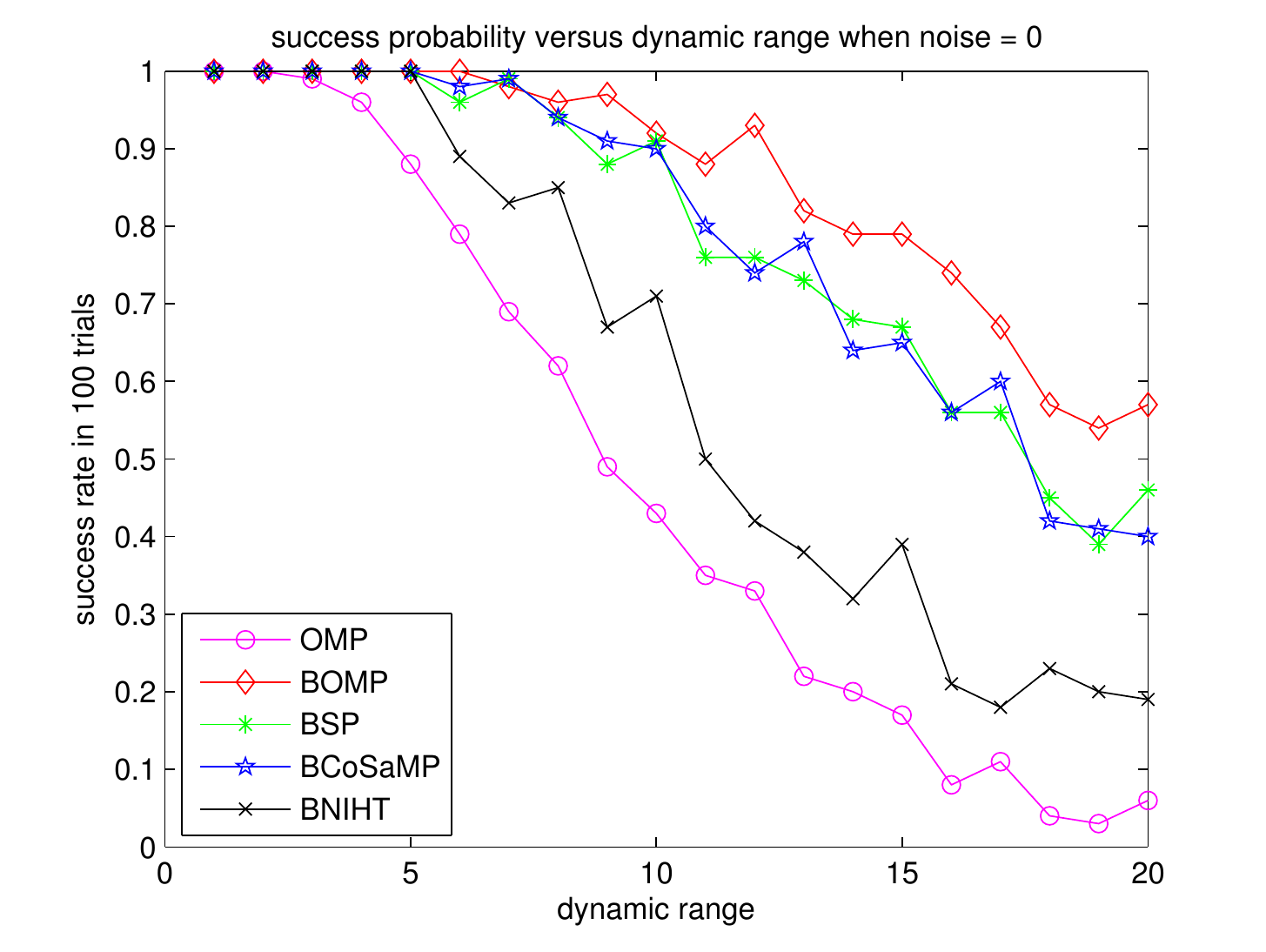}
\includegraphics[width=8cm,height=6cm]{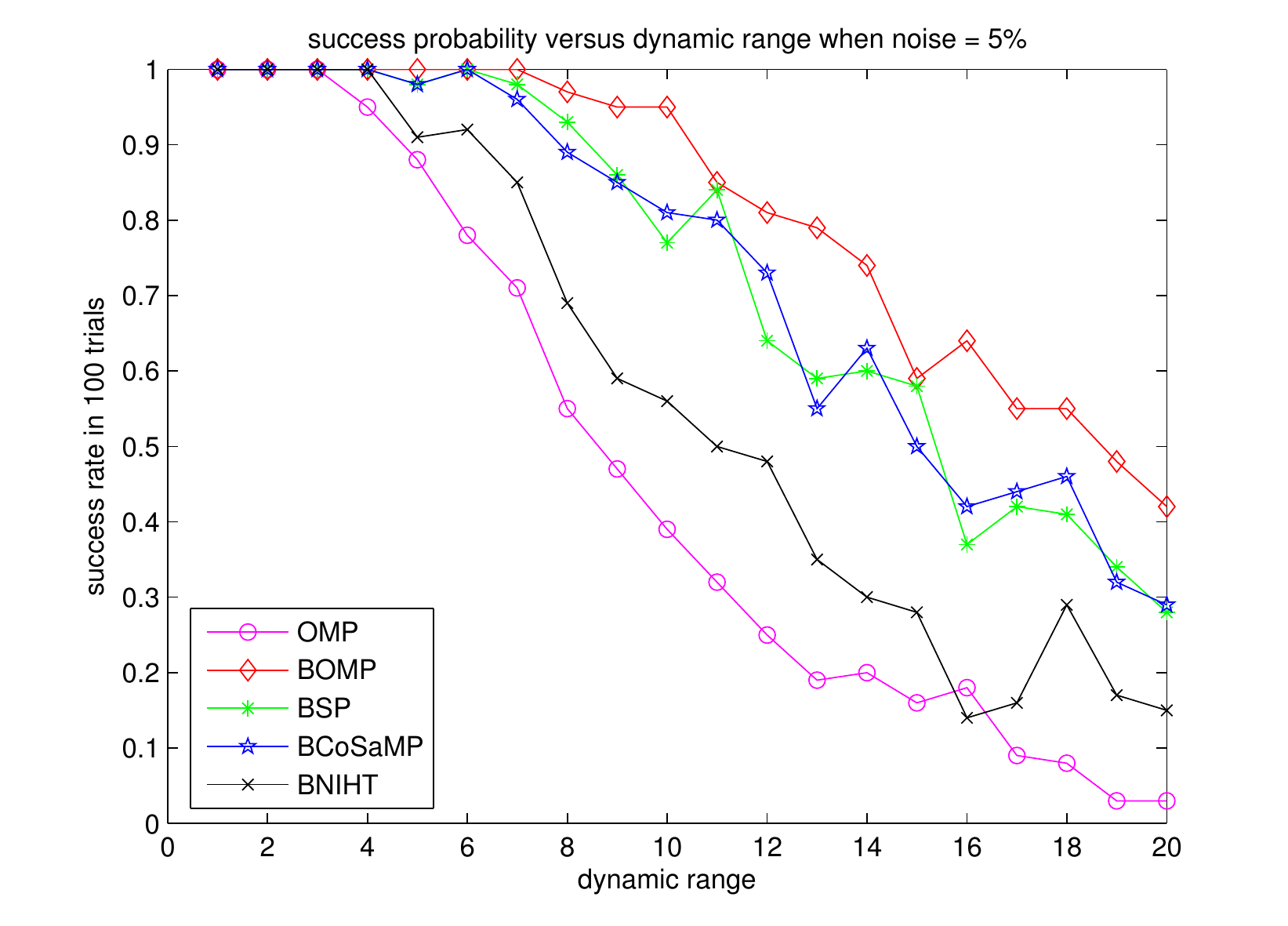}
\includegraphics[width=8cm]{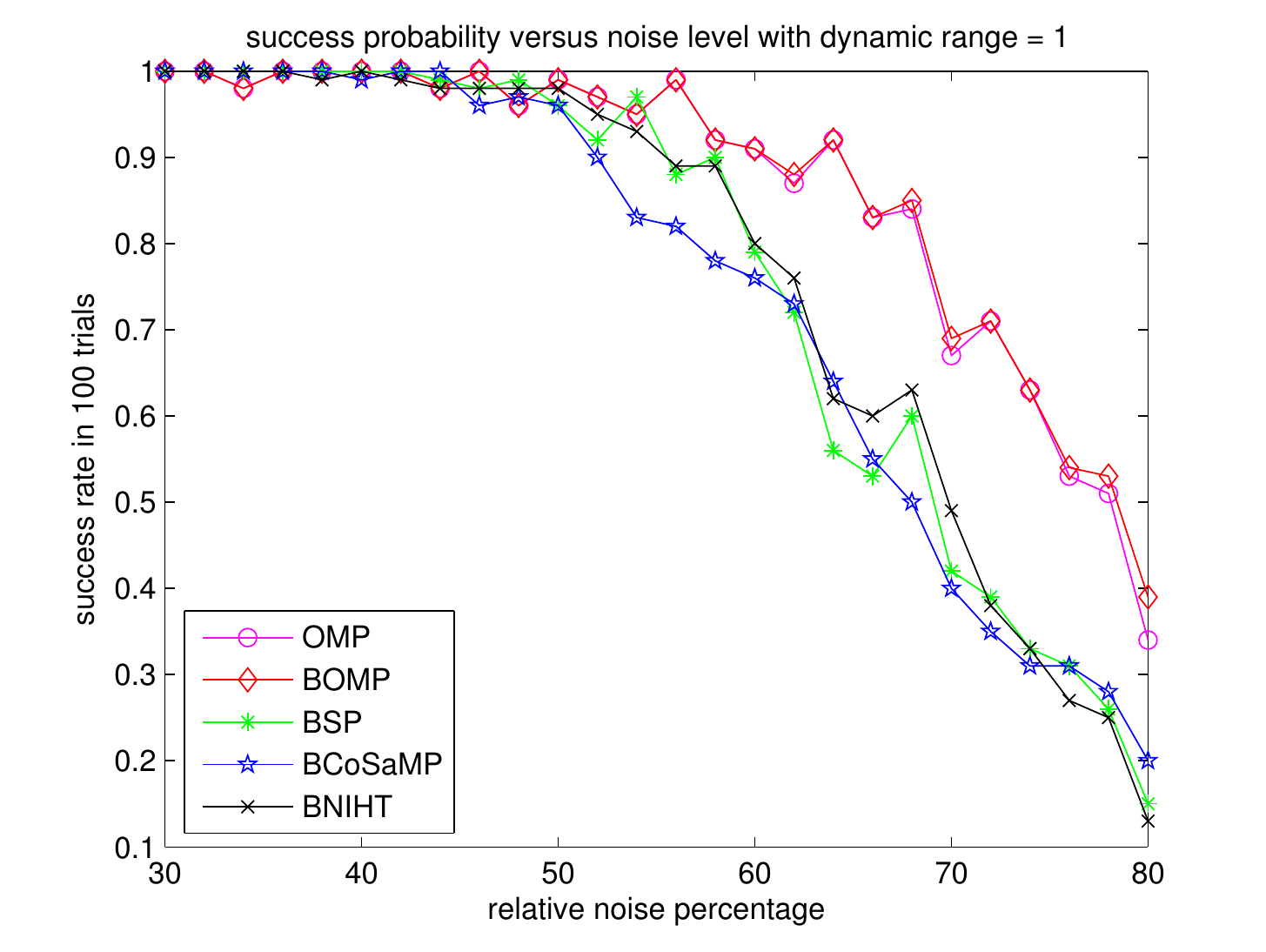}
\includegraphics[width=8cm, height=6cm]{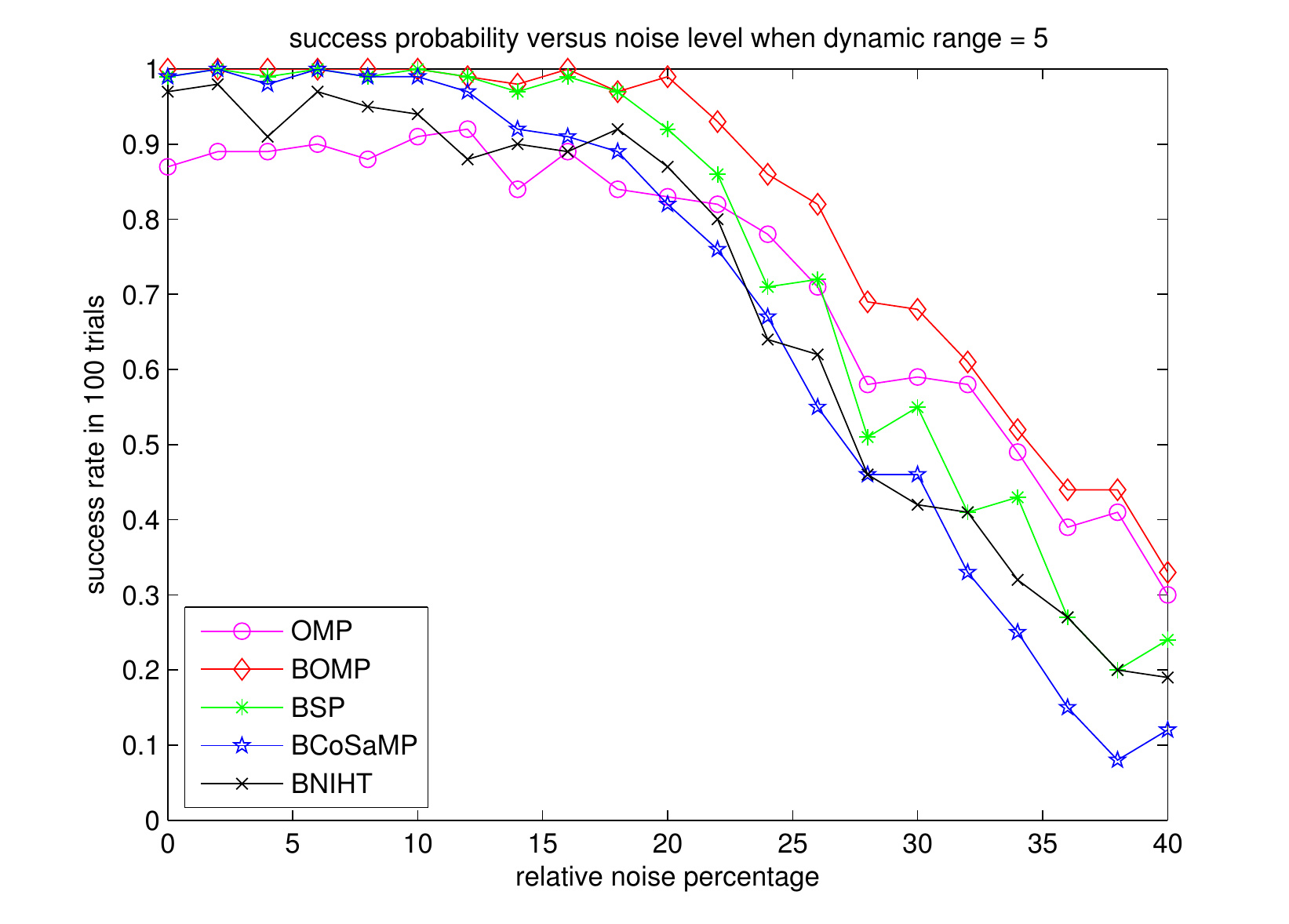}
\includegraphics[width=8cm,height=6cm]{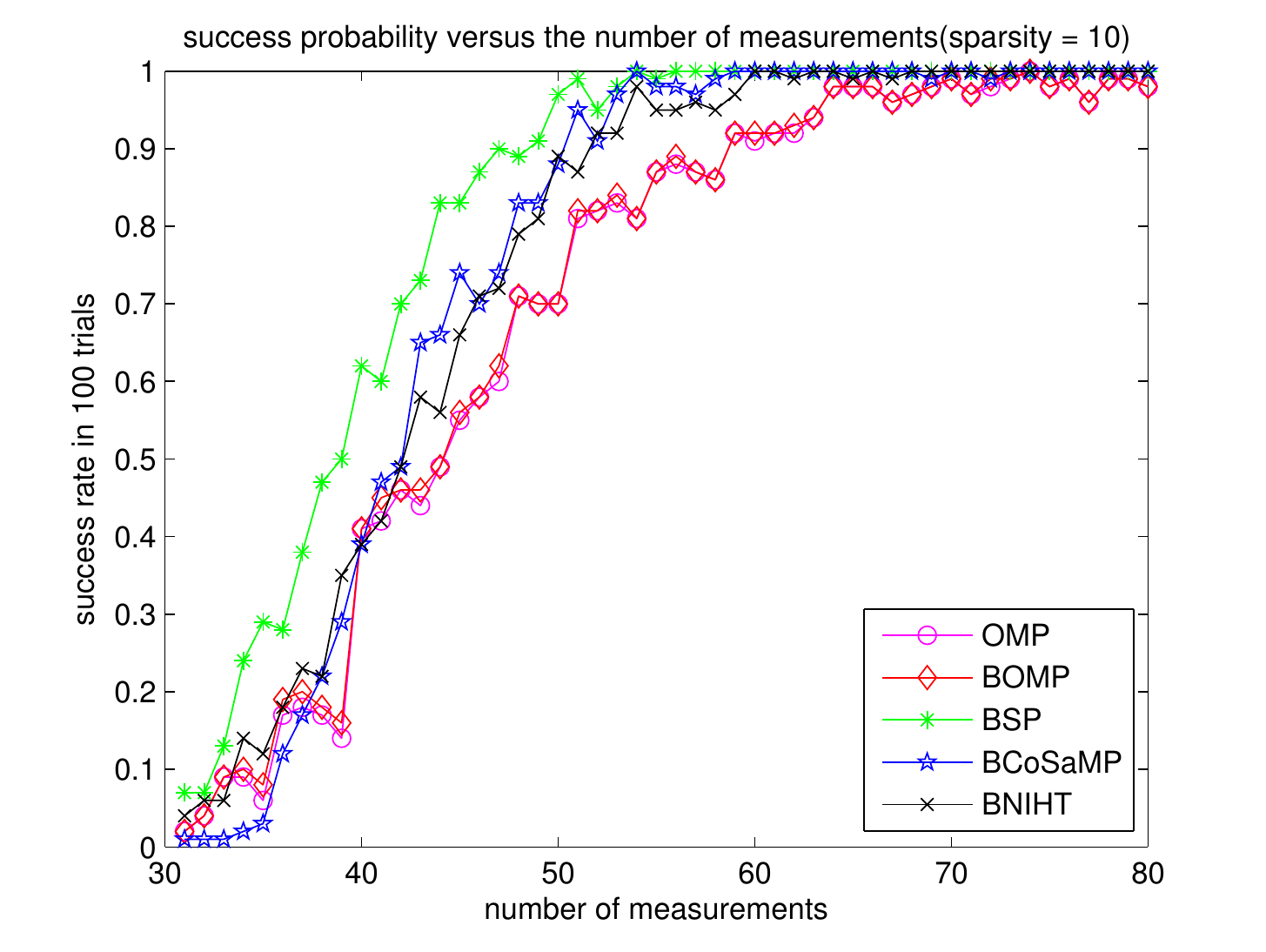}
\includegraphics[width=8cm,height=6cm]{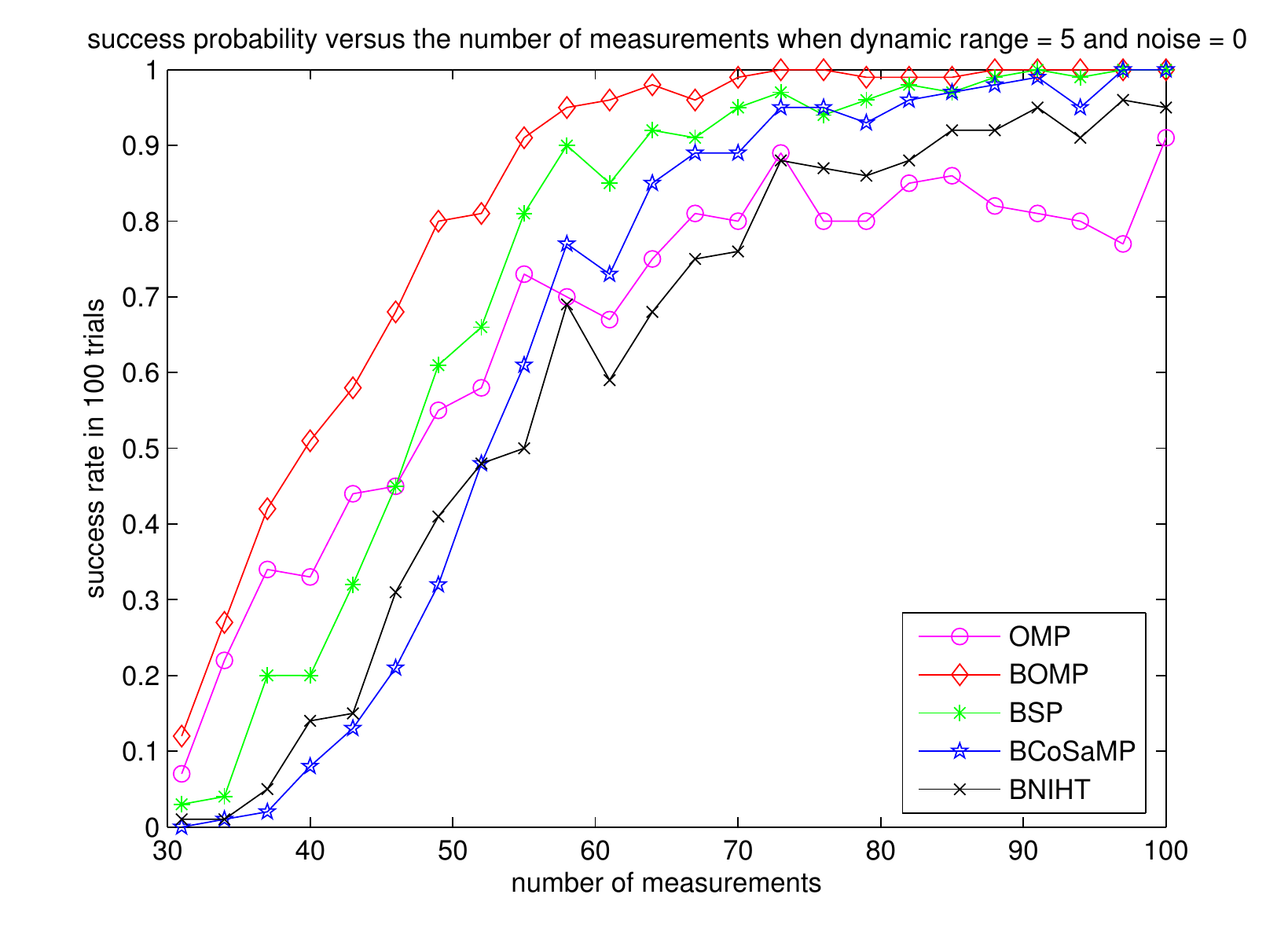}
  \caption{Success probability of various band-excluded algorithms versus dynamic range for $0\%$ (top left) and $5\%$ (top right) noise, versus noise level for dynamic range 1 (middle left) and 5 (middle right)
  and versus number of noiseless measurements for dynamic range 1 (bottom left) and 5 (bottom right).}
  \label{figB}
 \end{figure}

We test our various band-exclusion algorithms on the matrix
 \beq
 \label{10.5}
 \frac{e^{-2\pi i (l-1) \xi_k/F}}{\sqrt{N}},\quad k=1,...,N,\,\, l=1,...,RF
 \eeq
  where $\xi_k$ is uniformly and independently distributed in $(0,1)$. When $F=1$, 
 $\bA$ is the random partial Fourier matrix analyzed  in \cite{Rau} and,  with sufficient number of samples,   successful recovery with $\bA$  is guaranteed with high probability. 
 For large $F$, however, $\bA$ has a high mutual coherence. 
Unless otherwise stated, we use $N=100$, $M=4000$ and $F=20$
and  the i.i.d. Gaussian noise $\mbe\sim N(0,\sigma^2 I)$ in
our simulations. 
Recall  the  decay profile of pairwise coherence as the index separation increases in Figure $\ref{fig1}$ (right). 
 The $\eta$-coherence band is about  $0.7$ RL  in half width with
 $\eta=0.3$.   

We compare performance
 of various algorithms in terms of success probability versus 
  dynamic range,  noise level, number of measurements and resolution. 
 Unless otherwise stated, we use in our simulations 10 randomly phased and located  objects, separated by at least 3 RLs. 
 A reconstruction is counted as a success  if every reconstructed  object
is within 1 RL of 
the object support. This is equivalent to
the criterion that the Bottleneck distance between
the true support and the reconstructed support
is less than 1 RL. 

For two subsets $A$ and $B$ in $\IR^d$ of the same  cardinality,
the Bottleneck distance $d_B(A,B)$ is defined as follows.
Let $\cM$ be the collection of all one-to-one mappings  from $A$ to $B$. Then
\[
d_B(A,B)=\min_{f\in \cM}\max_{a\in A} |a-f(a)|.
\]
For subsets in one dimension, the Bottleneck distance can be
calculated easily. Let $A=\{a_1,\ldots,a_n\}$ and $B=\{b_1,\ldots, b_n\}$ be listed in the ascending order. 
Then 
\[
d_B(A,B)=\max_{j} |a_j-b_j|.
\]
In higher dimensions, however, it is more costly to compute
the Bottleneck distance \cite{EI, IV}. The Bottleneck distance
is a stricter metric than the Hausdorff distance which does
not require one-to-one correspondence between the two target
sets. 

In the first set of experiments, we test 
various greedy algorithms equipped with
the BE step (only).
This includes BOMP, Band-excluded 
Subspace Pursuit (BSP), Band-excluded CoSaMP (BCoSaMP) and Band excluded
Normalized Iterative Hard Thresholding (BNIHT).   For comparison, we also
show  the performance of OMP without BE.

As shown in Figure \ref{figB}, BOMP has the best performance with respect  to dynamic range, noise and, in the case of higher dynamic range ($\geq 3$, bottom right panel),  number of measurements.  In the case of dynamic range equal to 1, 
BSP is  the best performer in terms
of number of measurements  followed closely by BCoSaMP
and BNIHT (bottom panel).  In the case of  dynamic range equal to 1, BOMP and OMP have almost identical performance
with respect to noise (middle left) and number of measurements (bottom left). 
The performance of BSP and BCoSaMP, however, 
depends crucially on the BE step without which both SP
and CoSaMP fail catastrophically (not shown).


In the next set of experiments, we test BLO-equipped algorithms.  
For the purpose of comparison, we also test the algorithm, Locally Optimized OMP (LOOMP) which is the same as Algorithm 3
but without BE.

Lasso-BLOT
is implemented with the regularization parameter
\beq
\label{10.3}
\lambda=0.5\sqrt{\log{M}}
\eeq
or
\beq
\label{10.4}
\lambda=\sqrt{2\log{M}}
\eeq
which is proposed in \cite{CDS}. Other larger values have been proposed
in \cite{CP, CP2}. Our numerical experiments indicate
that for matrix (\ref{10.5}) with large $F$ the choice
(\ref{10.3}) is nearly optimal among all $\lambda/\sqrt{\log{M}}\leq 10$
and relative noise up to $5\%$. The superiority of the choice (\ref{10.3}) to (\ref{10.4}) (and other choices) manifests clearly across all performance figures  involving 
both of them.

\begin{figure}
\includegraphics[width=8cm]{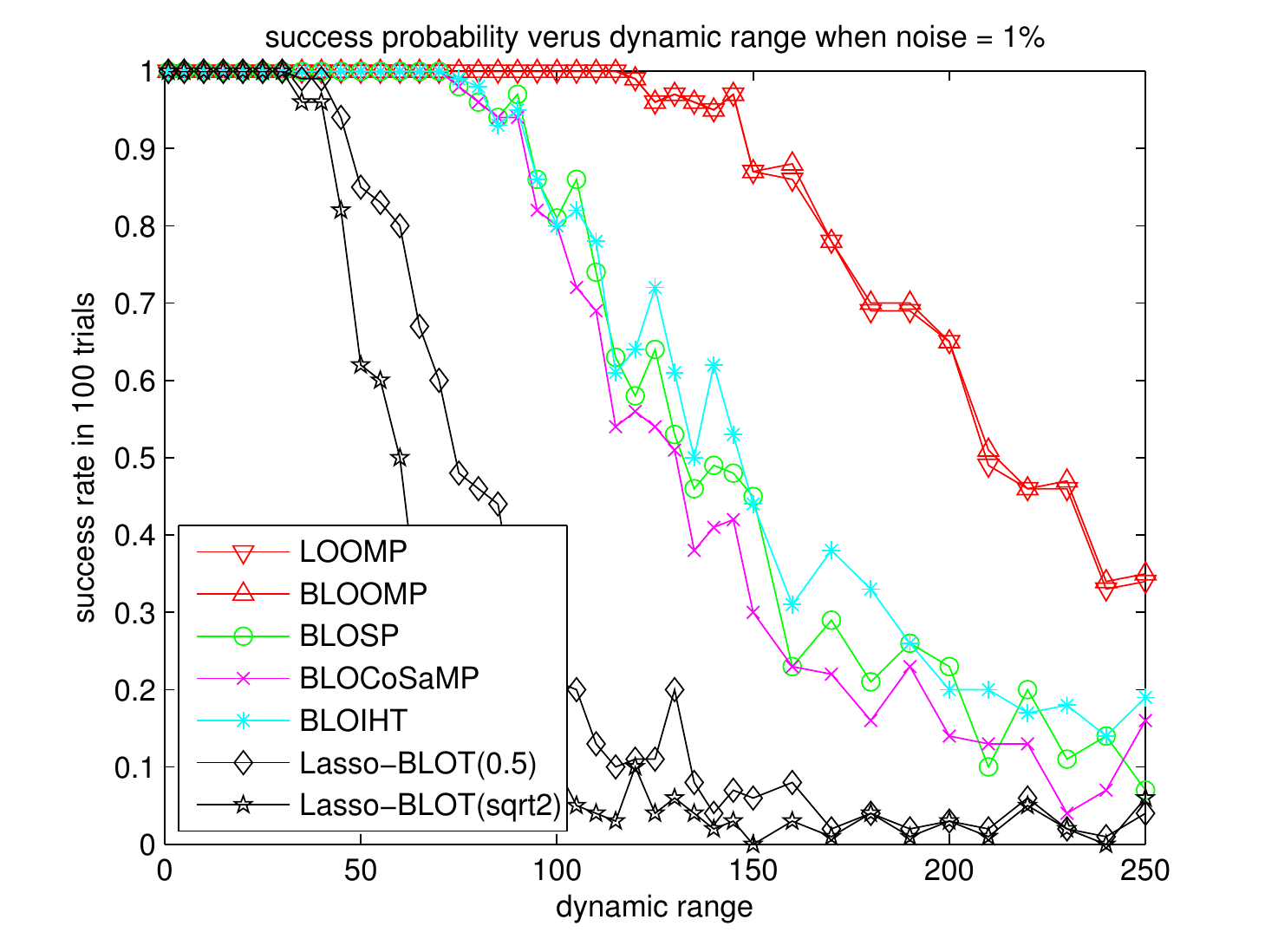}
\includegraphics[width=8cm]{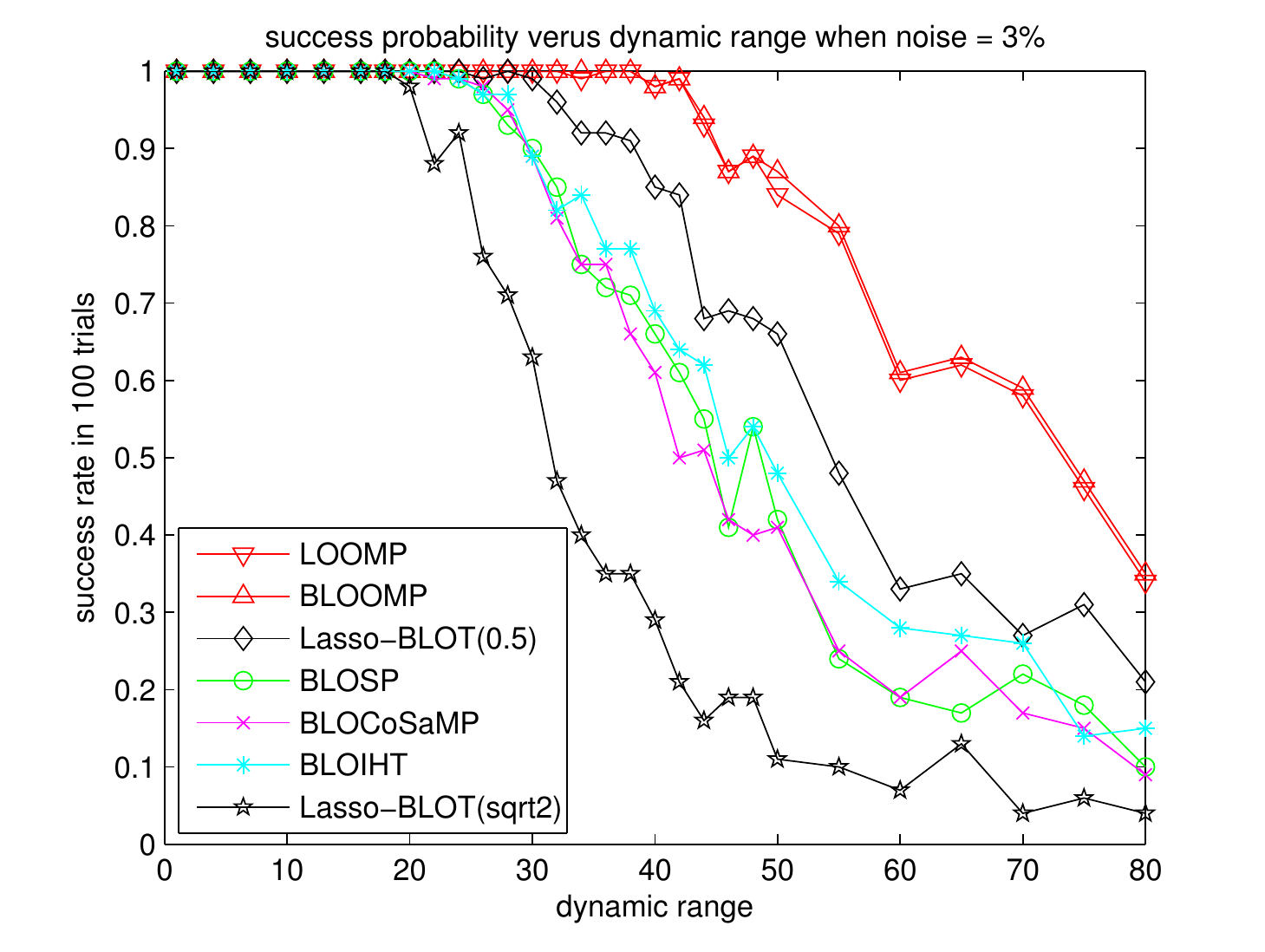}
\caption{Success probability versus dynamic range with
$1\%$ (left) and $3\%$ noise (right).}
\label{fig44}
\label{fig42}
\end{figure}
\begin{figure}
\includegraphics[width=8cm]{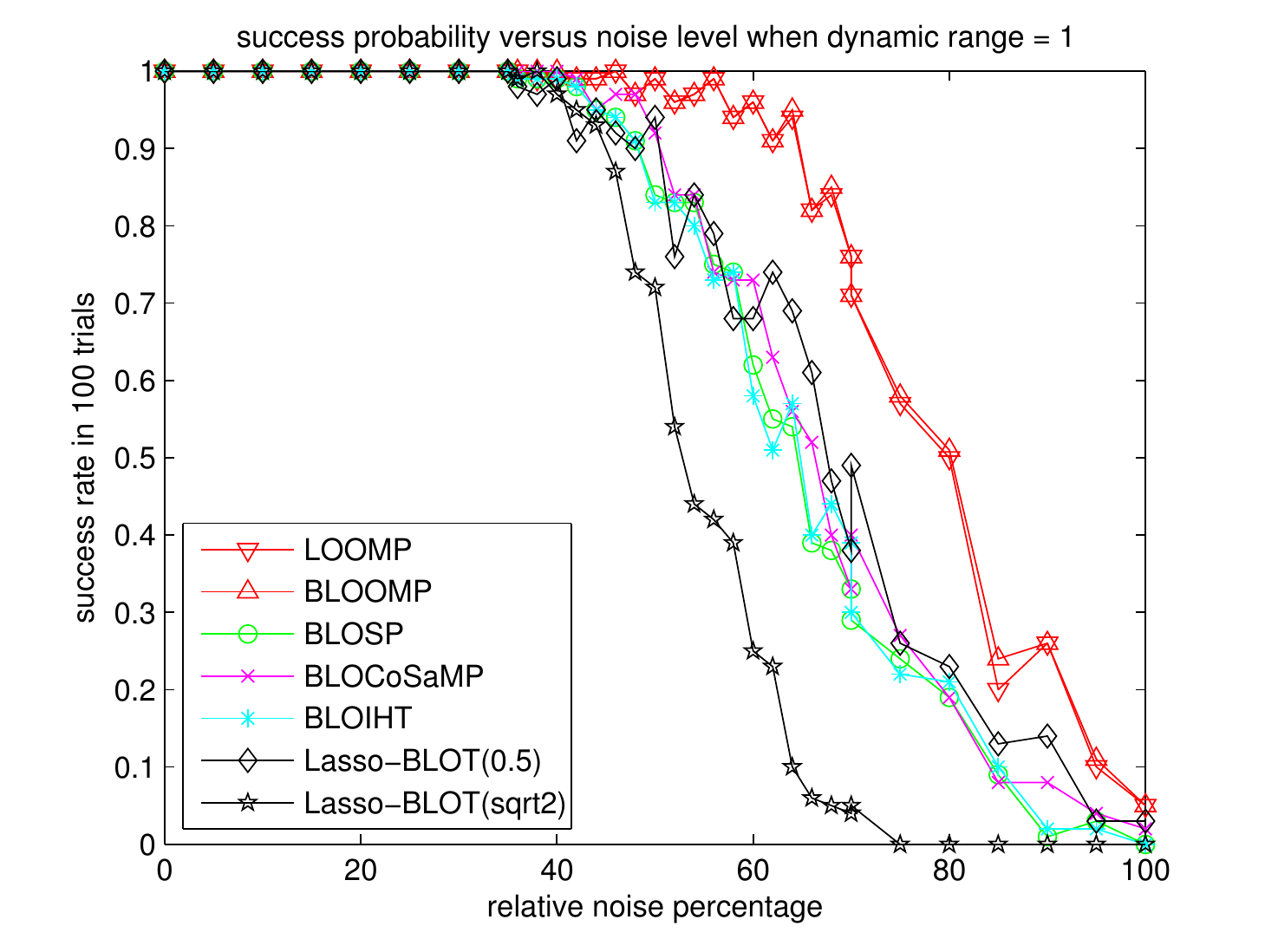}
\includegraphics[width=8cm]{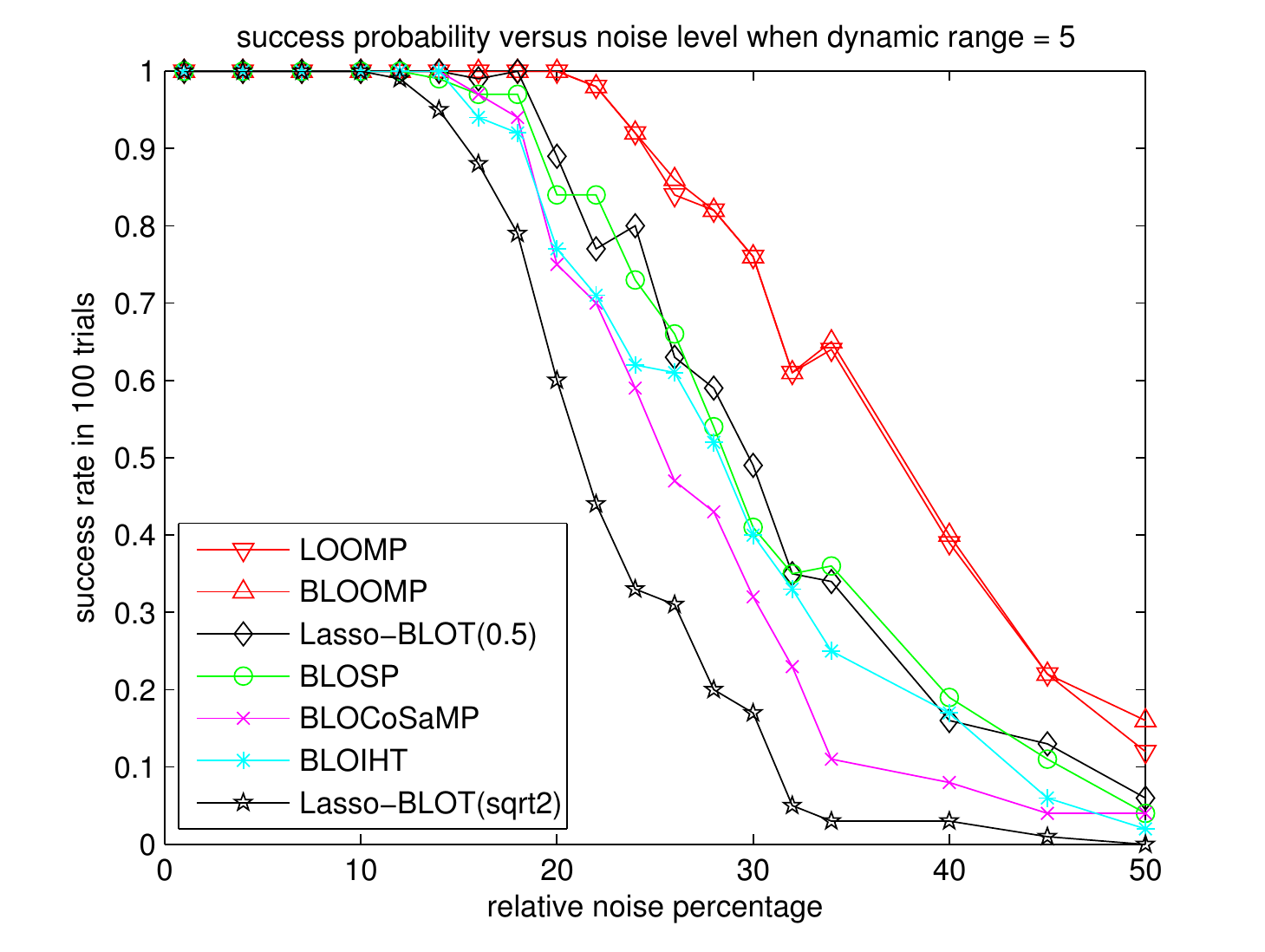}
\caption{Success probability versus relative noise level with
dynamic range 1 (top) and 5 (bottom). }
\label{fig49-1}
\end{figure}

FIgure \ref{fig42} shows success probability versus
dynamic range in the presence  of noise. The top performers are LOOMP and
BLOOMP both of which 
can handle large dynamic range.  In the noiseless case,
the success rate for LOOMP,
BLOOMP, BLOSP, BLOOMP and BLOCoSaMP stays near unity 
for dynamic range up to as high as $10^{14}$. 
With $1\%$ noise (left panel), BLOSP, BLOCoSaMP and BLOIHT 
perform better than Lasso-BLOT with either (\ref{10.3}) or
(\ref{10.4}) while 
with  $3\%$ noise (right panel), 
BLOSP, BLOCoSaMP and BLOIHT performance curves have dropped below
that of Lasso-BLOT with (\ref{10.3}). But the noise stability of 
Lasso-BLOT never catches up with that of LOOMP/BLOOMP even
as the noise level increases 
as can be seen in Figure \ref{fig49-1}.

 Figure \ref{fig49-1} shows
that  LOOMP and BLOOMP remain  the top performers with
respect to noise while Lasso-BLOT with (\ref{10.4})
has the worst performance. Lasso-BLOT with (\ref{10.3}), however,
is a close second in noise stability. As seen 
in Figures \ref{fig42} and \ref{fig49-1}, the performance of Lasso-BLOT
depends significantly on the choice of the regularization parameter. 

With respect to number of measurements (Figure \ref{fig411-1}), 
BP/Lasso-BLOT with (\ref{10.3}) is  the best performer,  followed closely by 
BLOSP and BLOCoSaMP  for dynamic range 1 (left panels) while
for  dynamic  range 10, BLOOMP and LOOMP perform  significantly better than the rest (right panels). As clear from
the comparison of the top left and right panels of Figure \ref{fig411-1}, 
at  low level of noise
the performance of BLOOMP and LOOMP improves significantly
as the dynamic range increases from 1 to 10. In the meantime,  
 the performance of BLOSP, BCoSaMP and BLOIHT improves
slightly while  the performance of BP/Lasso-BLOT deteriorates. At $10\%$ noise, however,  the performance of BLOOMP and LOOMP
is roughly unchanged as dynamic range increases while the performance of 
all other algorithms deteriorates significantly (bottom right). 

\commentout{
Comparing the performance of BLOIHT in Figures \ref{fig42}, 
\ref{fig49-1} and \ref{fig411-1} and the performance of BIHT
in Figure \ref{figB}, we see how much improvement has
been gained by the local optimization step. 
On the other hand,  
embedding  the LO step in  BIHT-BMT and BNIHT  
does not
further improve performance over that of  BLOIHT. 
}

 \begin{figure}
\commentout{
\includegraphics[width=8cm,height=6cm]{May15/SuccessMeasurementRange1Noise0.eps}\\
\includegraphics[width=8cm,height=6cm]{May8/SuccessMeasurementRange1Noise0.eps}
\includegraphics[width=8cm,height=6cm]{May8/SuccessMeasurementRange1Noise5.eps}
\includegraphics[width=8cm,height=6cm]{May8/SuccessMeasurementRange5Noise1.eps}
\includegraphics[width=8cm,height=6cm]{May8/SuccessMeasurementRange20Noise0.eps}
}
\includegraphics[width=17cm]{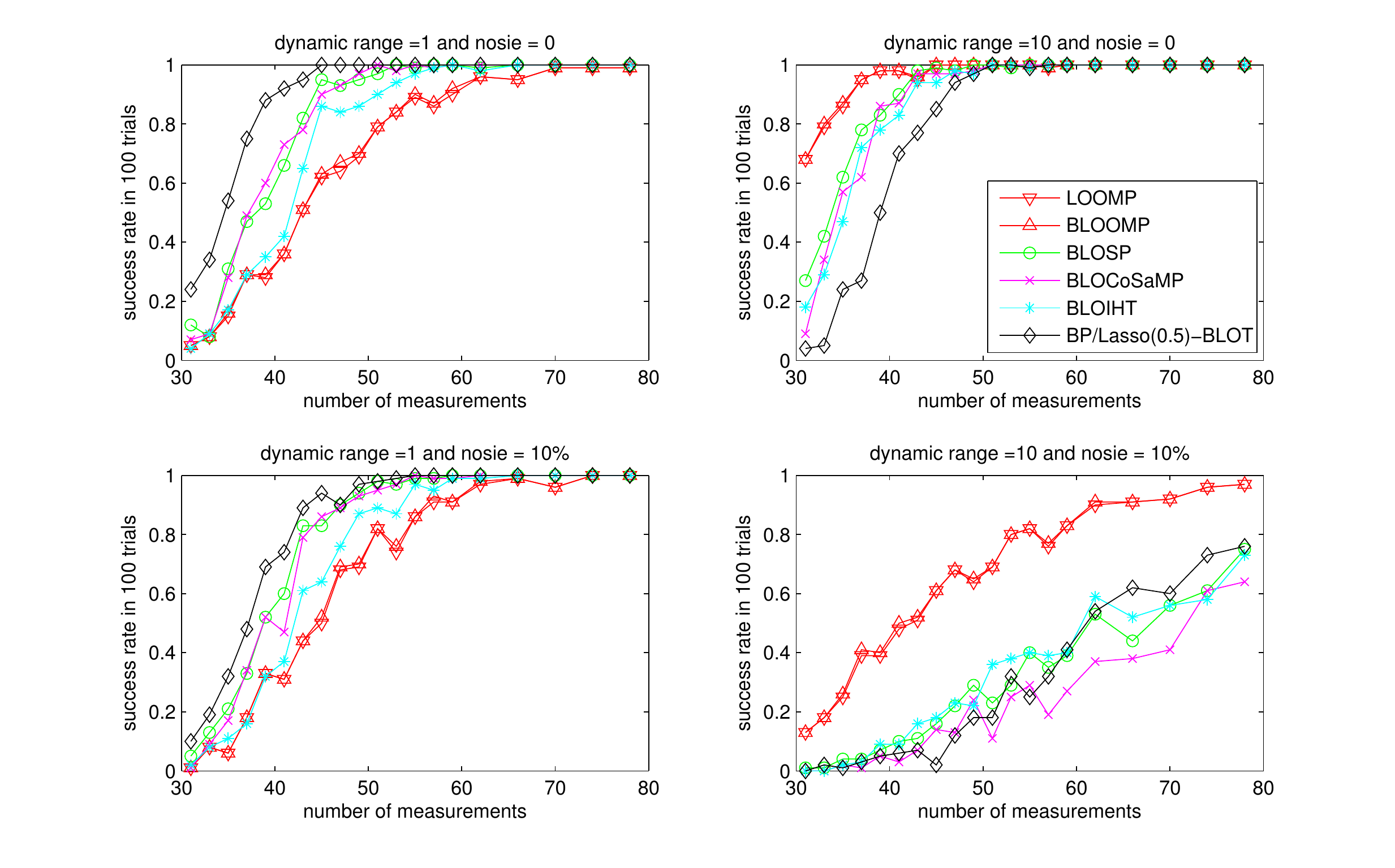}
\caption{Success probability versus the number of measurements with dynamic range 1, $0\%$ noise (top, left),
dynamic range 1, $10\%$ noise (bottom, left), dynamic range 10, $0\%$ noise (top, right)
and with dynamic range 10, $10\%$ noise (bottom, right).}
\label{fig411-1}
\end{figure}

Next we compare   the resolution performance 
of the various algorithms for 10 randomly phased objects of unit  dynamic range in the absence of noise. The 10 objects are consecutively   located  and separated by  equal length varying from $0.1$ to 3 RLs. The whole object support is, however, randomly 
shifted for each of the 100  trials.  For closely  spaced
objects, it is necessary to modify the band exclusion and local optimization rules:  If $h$ is the object spacing, 
we use $h/2$ to replace 2 RLs of the original BE rule
and 1 RL of  the original LO rule.

\begin{figure}
\includegraphics[width=8cm]{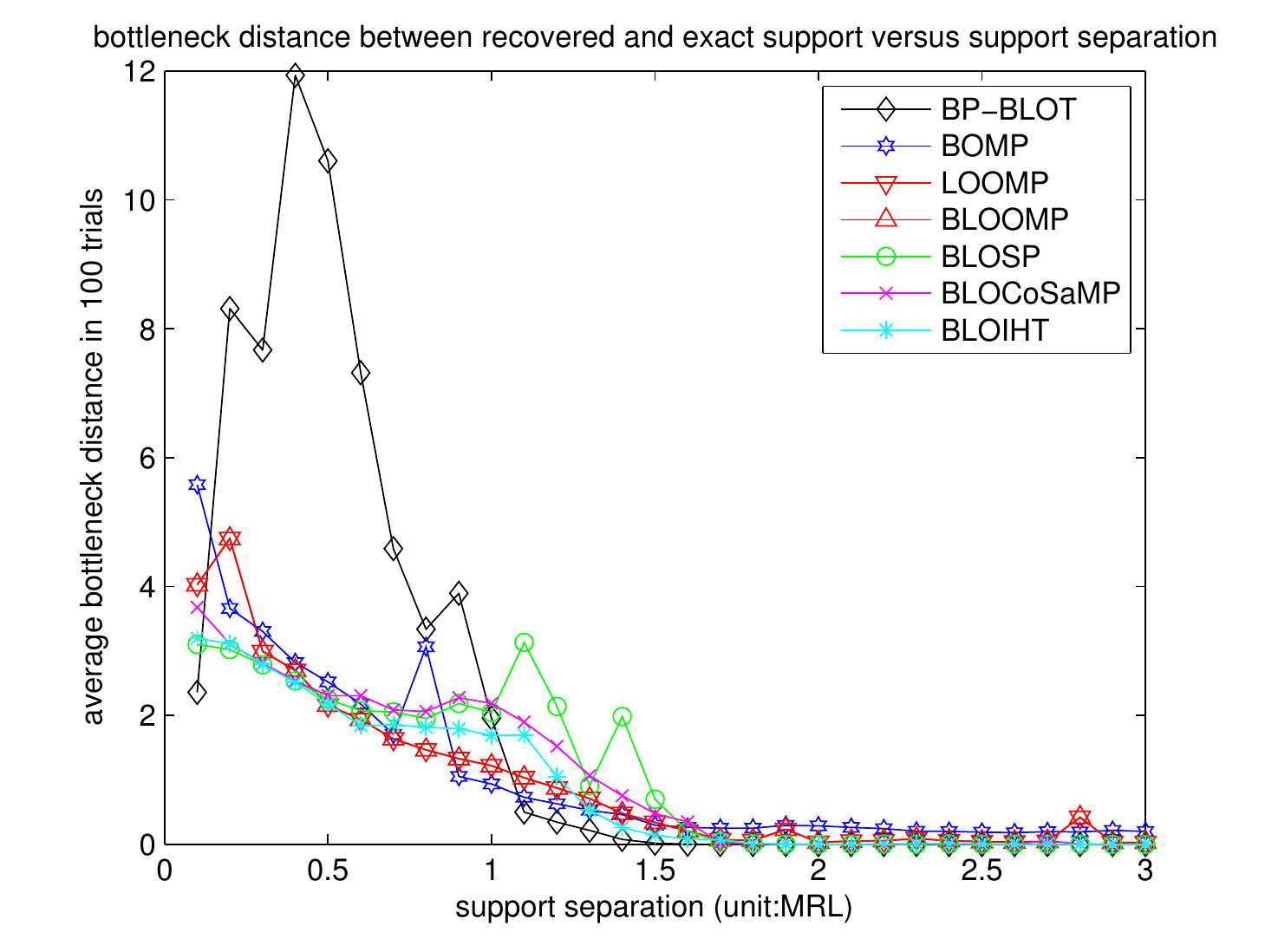}
\includegraphics[width=8cm]{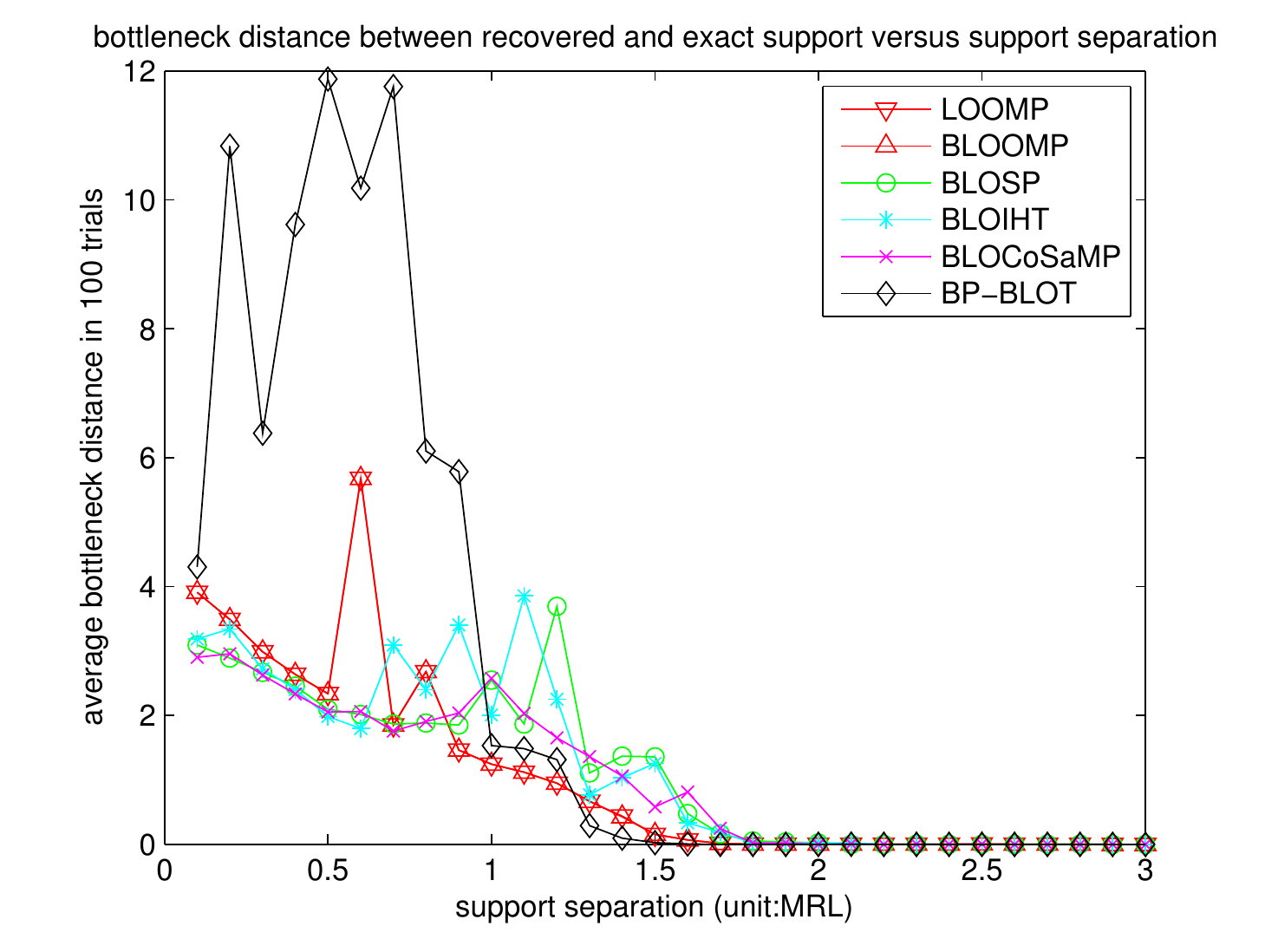}
\includegraphics[width=8cm]{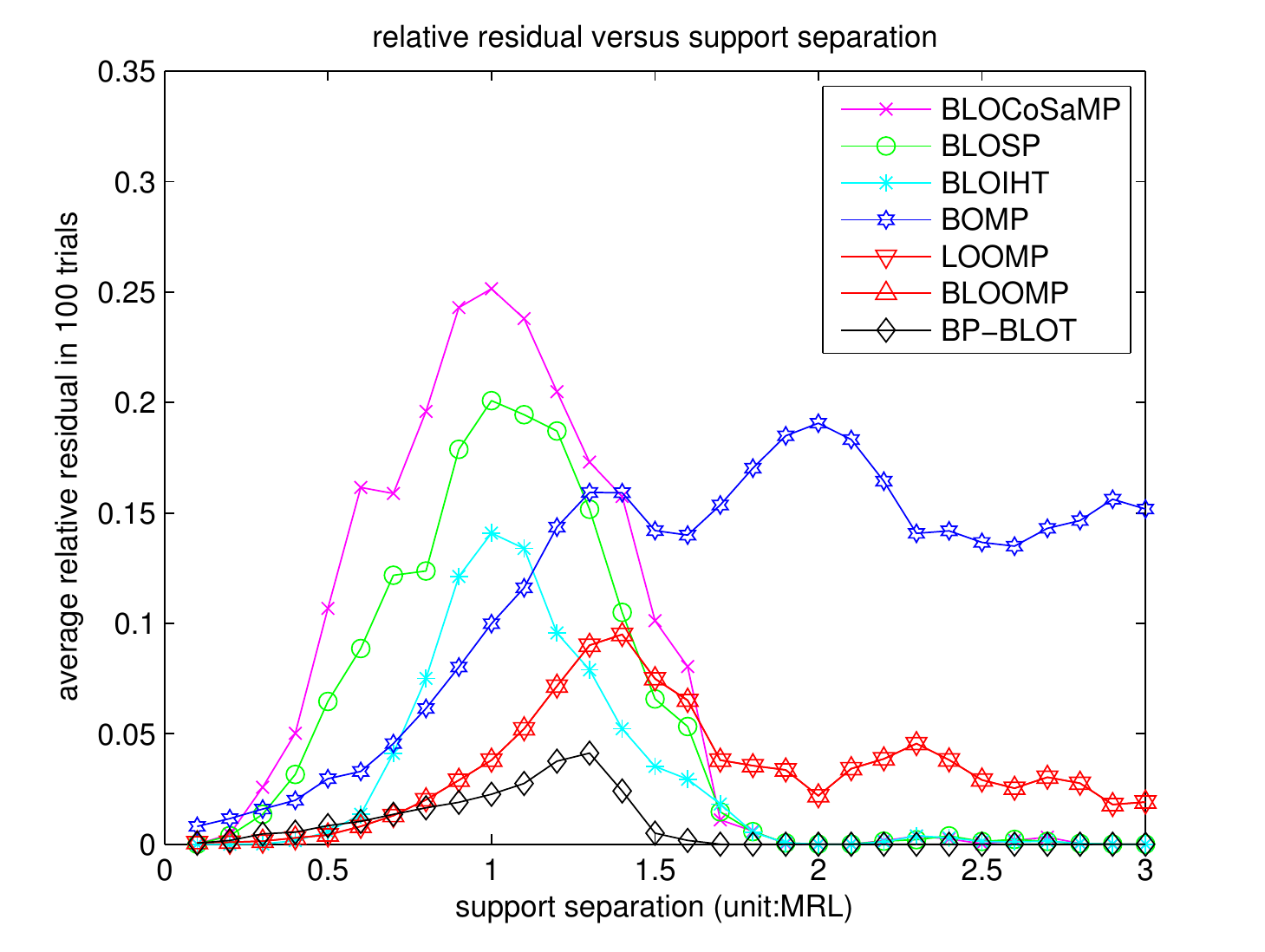}
\includegraphics[width=8cm]{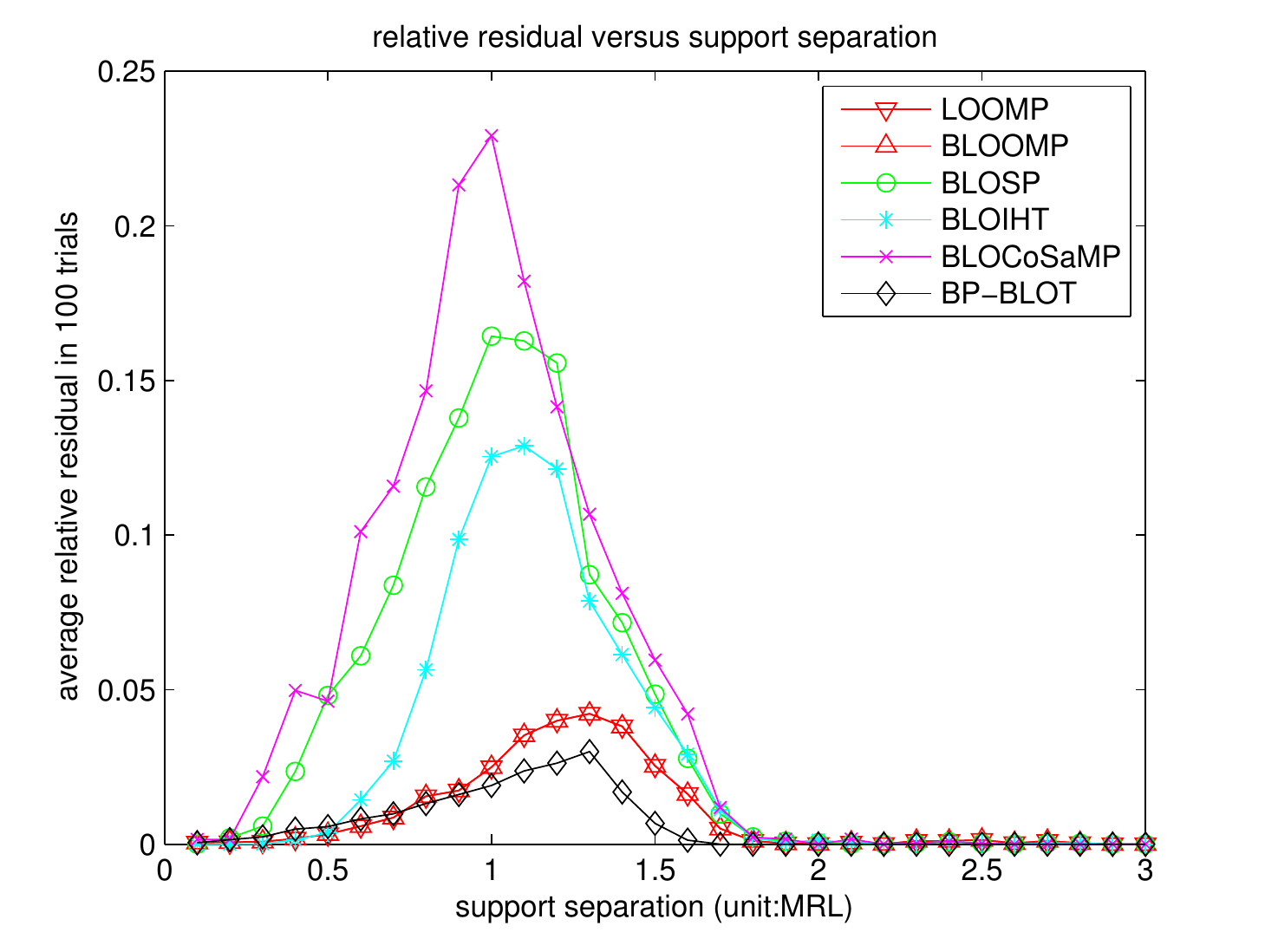}
  \caption{The average Bottleneck distance for dynamic range 1 (top left) and 10 (top right), and the
  relative residual for dynamic range 1 (bottom left) and 10 (bottom right)  versus separation of
  objects}
  \label{fig422}
\end{figure}

 Figure \ref{fig422} shows 
the averaged Bottleneck distance between
the reconstruction and the true object support (top panels) and
the residual (bottom panels) as a function of
the object spacing.  For this class of objects, 
 BP-BLOT has the best resolution for dynamic range up to 10  followed closely  by
BLOIHT for dynamic range 1 and by BLOOMP/LOOMP for dynamic
range 10. The {\em high precision} (i.e. nearly zero Bottleneck distance) resolution ranges from about 1.5 RLs for  BP-BLOT to  about 1.7 RLs for the rest. Consistent with what we have
seen in Figure \ref{fig411-1}, the resolving power of BLOOMP/LOOMP
improves significantly as the dynamic range increases while that of
BP-BLOT deteriorates. Note that in the case of unit dynamic range,
BOMP recovers the support as well as BLOOMP/LOOMP does.
BOMP, however, produces a high level of residual  error 
even when objects are widely separated.  There is
essentially no difference between the resolving power of BLOOMP
and LOOMP. 

It is noteworthy that the relative residuals of  all tested algorithms peak at separation  between  1 and 1.5 RLs  and decline  to zero as 
the object separation decreases. In contrast, the average Bottleneck distances  increase as the separation decreases
except for BP-BLOT when  the separation drops below
0.5 RL. When the separation drops below 1 RL,  the Bottleneck
distance between the objects and the reconstruction indicates
that the objects are not well recovered by any of the
algorithms (top panels). The vanishing residuals in this regime  indicates  nonuniqueness of 
sparse solutions.

Figures \ref{fig44}-\ref{fig422} show negligible  difference between the performances of LOOMP and BLOOMP.  
To investigate their distinction more closely, we test 
the stability with respect to the gridding error for
various $F$'s (cf. Figure \ref{fig328-1}).  We consider 
randomly phased objects of dynamic range 10 that 
are randomly located in $[0,1000)$ and separated by at least $3$ RLs.
We compute the reconstruction error in the Bottleneck distance and $L^2$-norm averaged over 100 trials with $5\%$ external noise
and show the result in  Figure \ref{fig328-2}. 
Evidently the advantage of BLOOMP over LOOMP 
lies in the cases when  the refinement factor $F$ is less
than 10 and the gridding error is sufficiently large. When $F\geq 10$, the difference between their performances  with respect to gridding error 
 is  negligible. On the other hand, for $F=5$ both
 BLOOMP and  LOOMP's reconstructions  would have
 been considered a failure given the magnitudes of error
 in the Bottleneck distance. 




\commentout{
\begin{figure}
\includegraphics[width=15cm, height=8cm]{March26/LOOMPandLOBOMP.eps}
\caption{The red asterisks represent the absolute values of  randomly phased objects of dynamic range 20 and the blue
stars are the reconstructions by LOOPM (left) and BLOOMP (right) with $5\%$ noise. }
\label{fig47-1}
\end{figure}
}

\begin{figure}
\includegraphics[width=8cm]{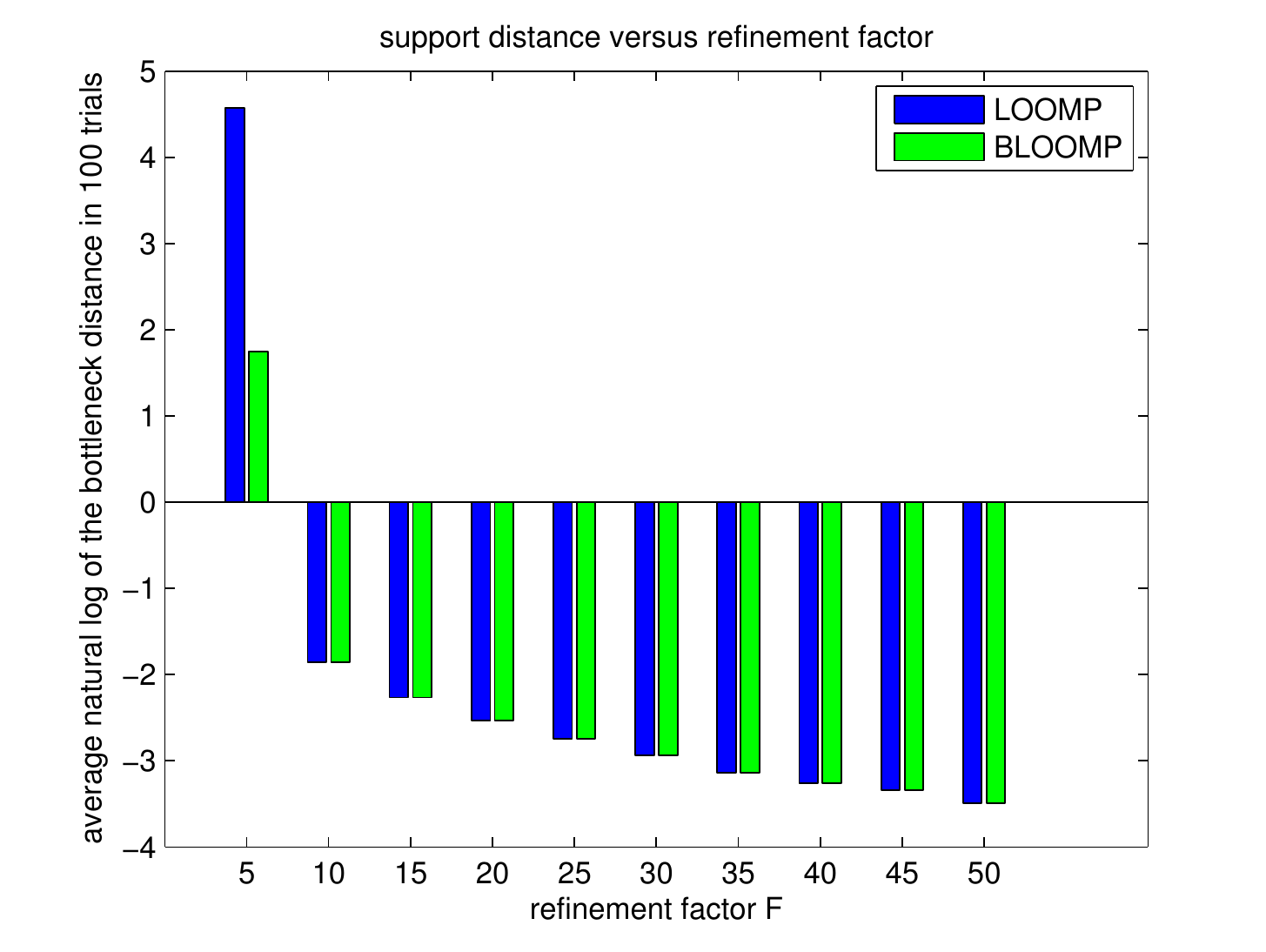}
\includegraphics[width=8cm]{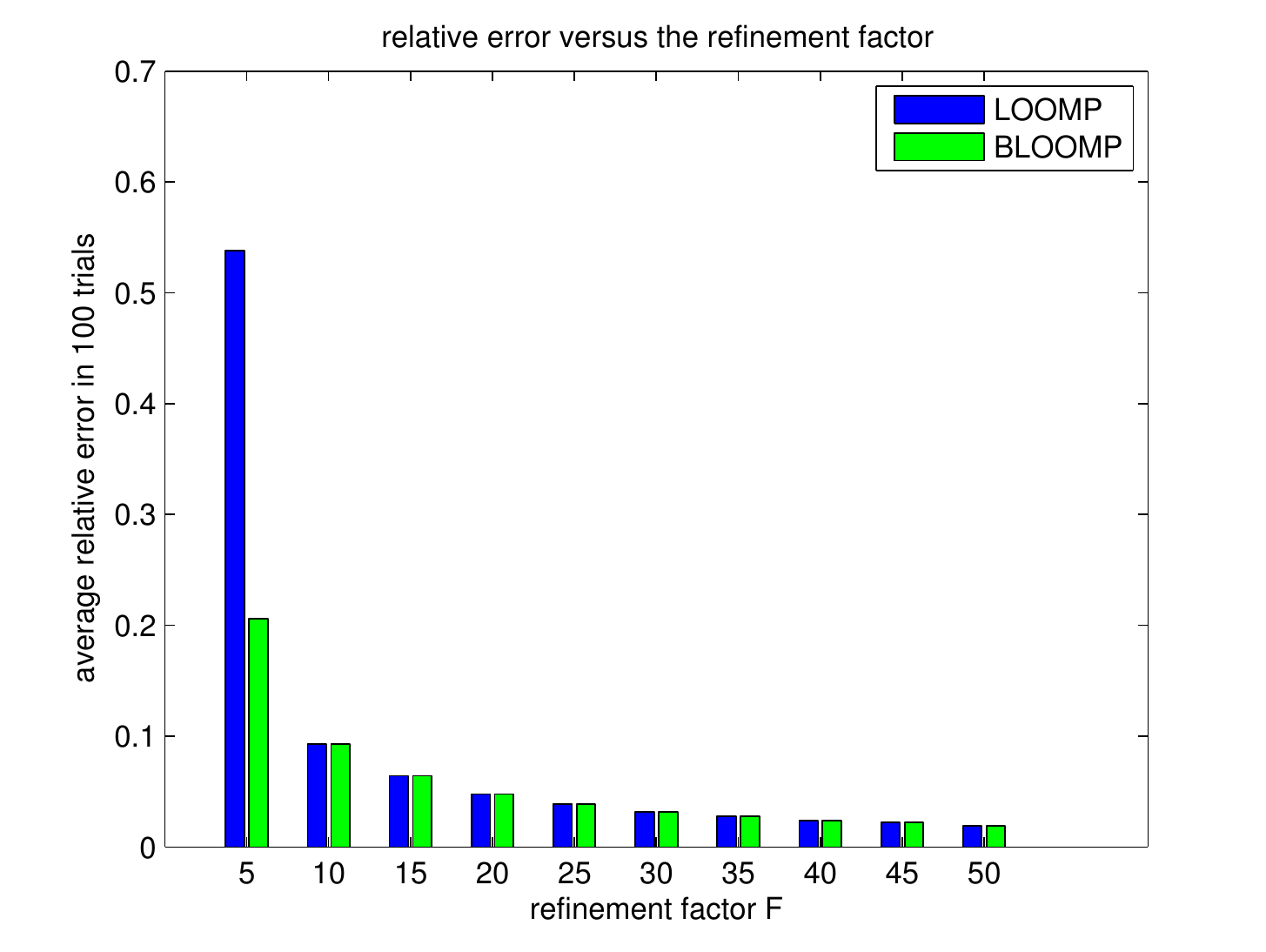}
\caption{The reconstruction error of LOOMP and  BLOOMP as measured by the Bottleneck
distance (left, semi-log) and  the relative $L^2$-norm (right)
for 10 objects of dynamic range 10
as a function of $F$.}
\label{fig328-2}
\end{figure}

\section{Comparison with other algorithms in the  literature}\label{sec6}
The present work is inspired by the performance guarantee established in \cite{music} that
the MUSIC algorithm
aided by BMT 
produces  a support  estimate that is within 1 RL of  the locations of
sufficiently separated objects.

In comparison to  other CS literature on  coherent and redundant dictionary, our work resembles those of  \cite{CEN}
and \cite{DB}, albeit with a different   perspective. In fact, the algorithms developed
in \cite{CEN} and \cite{DB} can not
be applied to the spectral estimation problem formulated
in the Introduction. 
The algorithms developed here, however, can be
applied to their frame-based setting and this is what
we will do below for the purpose of comparison.

Following \cite{DB} we consider the following problem  
\beq
\label{101}
\bb= \bPhi \by + \mbe
\eeq
where $\bPhi$ is a $N \times R$ i.i.d Gaussian matrix of mean $0$ and variance $\sigma^2$.  
The signal to be recovered is given by $\by =\bPsi \mbx$ where $\bPsi$ is the  over-sampled, redundant DFT frame
\beq
\label{102}
\Psi_{k,j} = \frac{1}{\sqrt{R}} e^{-2\pi i \frac{(k-1)(j-1)}{RF}},\quad k=1,...,R,\quad j=1,...,RF.
\eeq
As before $F$ is  the refinement factor. Combining
(\ref{101}) and (\ref{102}) we have the same  form (\ref{linearsystem}) with $
\bA=\bPhi\bPsi
$ whose coherence pattern is shown in Figure \ref{fig10} similar to
Figure \ref{fig1}.

In the simulation, we take $N=100, R = 200, F = 20$ and $\sigma = \frac{1}{\sqrt{N}}$ so that $\bA\in \IC^{100\times 4000}$ as before.  We use randomly located and phased $\mbx=(x_j)$ which are  well separated in the sense  that $|x_j - x_k| \geq 3, \forall j \neq k$.

 The algorithm proposed in \cite{DB}, Spectral Iterative Hard Thresholding (SIHT),  relies on a measurement  matrix $\bPhi$ satisfying 
some form of RIP  and so does 
 the frame-adapted BP  proposed in \cite{CEN} 
 \beq
 \label{112}
 \text{min} \|\Psi^{\star}\bz \|_1 \qquad \text{ s.t } \|\Phi \bz -\bb \|_2 \le \|\mbe\|_2
\eeq
which, in addition, requires the analysis coefficients $\bPsi^*\by$ to be $s$-sparse
or $s$-compressible. 
The RIP is satisfied by i.i.d. Gaussian matrices of proper sizes. Their approaches are not  applicable, however, when the sensing matrix $\bA$ can not be
decomposed into a product $\bPhi\bPsi$ where $\bPhi$ 
has RIP. 

\begin{figure}
\includegraphics[width=8cm]{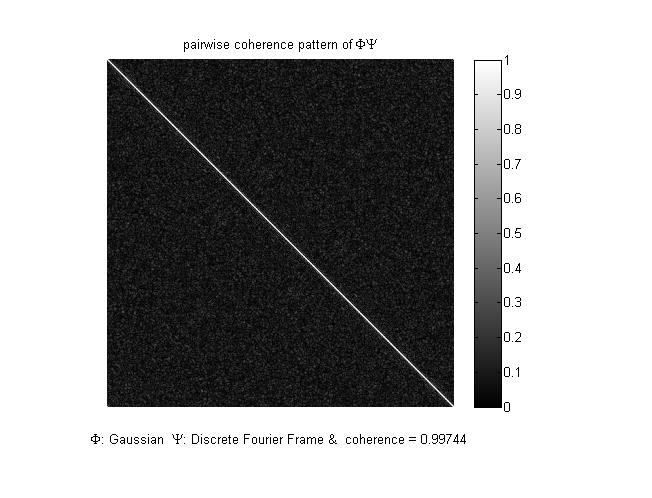}
\includegraphics[width=8cm]{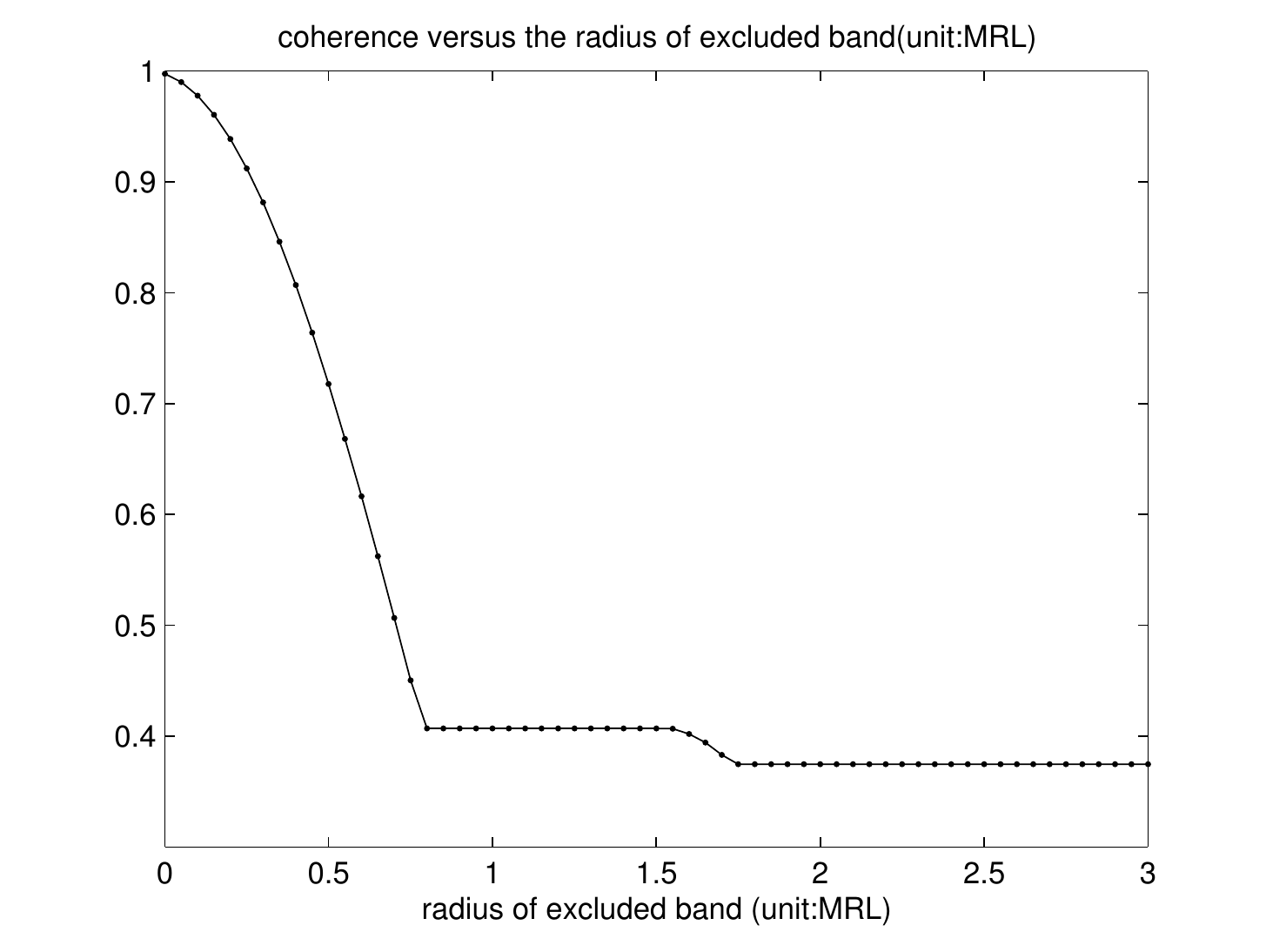}
\caption{The coherence pattern of $\bA=\bPhi\bPsi$}
\label{fig10}
\end{figure}

In the language of digital signal processing  (\ref{112}) is  a  $L^1$-analysis method 
while the standard BP or Lasso and the BLOT-version are {$L^1$-synthesis} methods. Both methods are based on the same
principle of representational sparsity.  In principle, the synthesis approach (such as the BP, Lasso and all BLO-based algorithms) is more general than the analysis approach since every analysis method can be
recast as a synthesis method while many synthesis formulations 
have no equivalent analysis form. 
In practice, however, each of them performs the best on different
types of signals 
\cite{EMR}.

For example, the bottom right panel of Figure \ref{fig19}  shows the absolute values of
the component of the vector  $\bPsi^*\by$ in
the order of descending magnitude. Clearly $\bPsi^*\by$ is
neither sparse nor  compressible. The shape of the curve can
be roughly understood as follows.  The DFT frame $\bPsi$ has a coherence band of roughly 1.5 RLs, corresponding to 30 columns, 
and since the 10 components in $\mbx$ are widely separated  $\bPsi^*\by$ has about $300$ significant components. 
The long tail of the curve is due to the fact that the pairwise
coherence of  
$\bPsi$ decays slowly  as the separation increases. 
Therefore the analysis approach (\ref{112}) would require a far higher number of 
measurements than 100  for accurate reconstruction. 

For the analysis approach, the main quantity of interest is $\by$. So 
in our comparison experiments, we measure
the performance  by the relative error $\|\hat\by -\by\|_2 / \|\by\|_2$ 
as in \cite{CEN,DB}.
We compute  the averaged relative errors as  dynamic range, noise level and the number of measurements vary.   In each of the 100 trials, 10 randomly phased and located 
objects (i.e. $\mbx$) of dynamic range 1 and i.i.d.   Gaussian $\Phi$ are generated. 

\begin{figure}
\includegraphics[width=8cm]{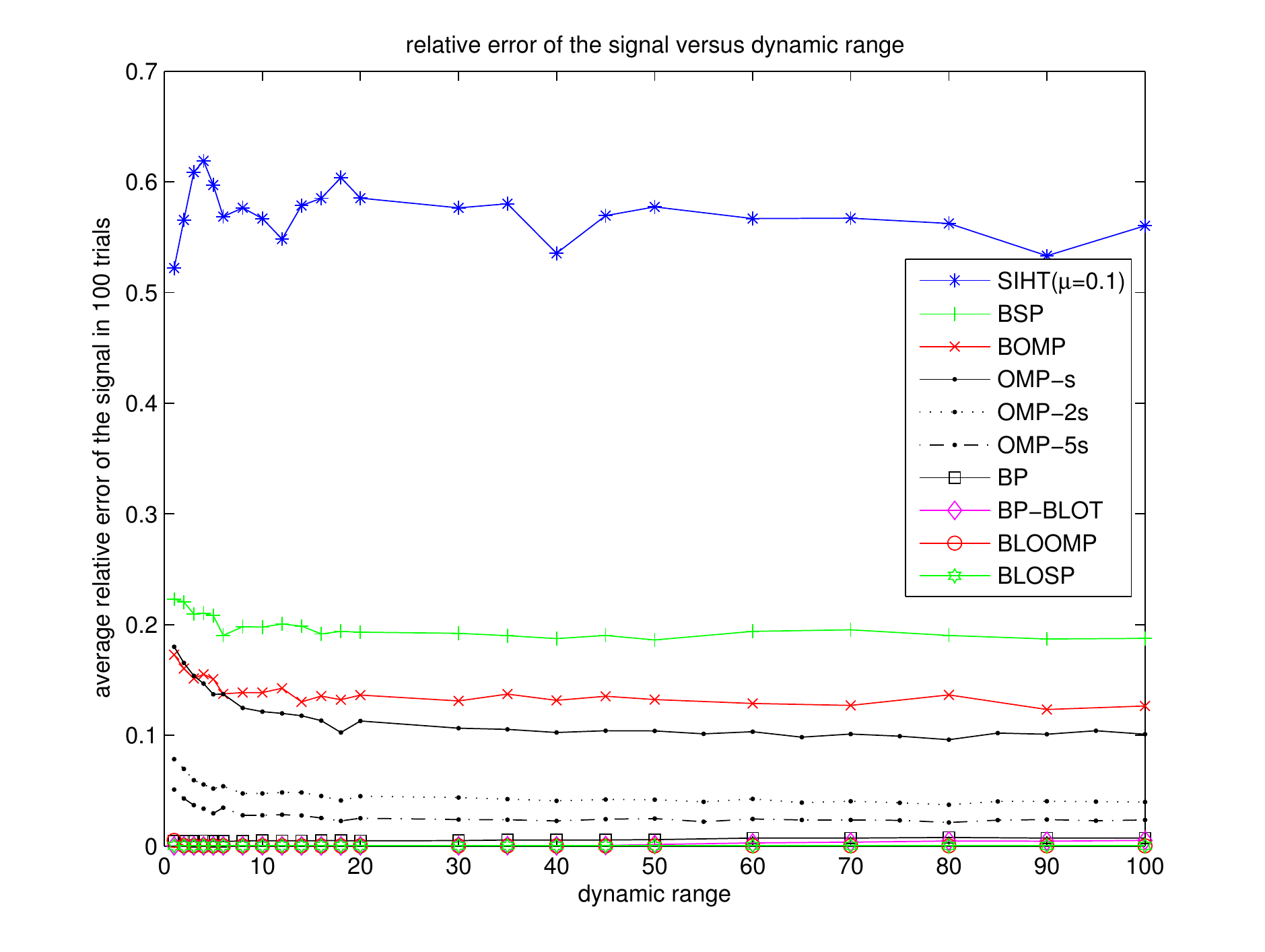}
\includegraphics[width=8cm]{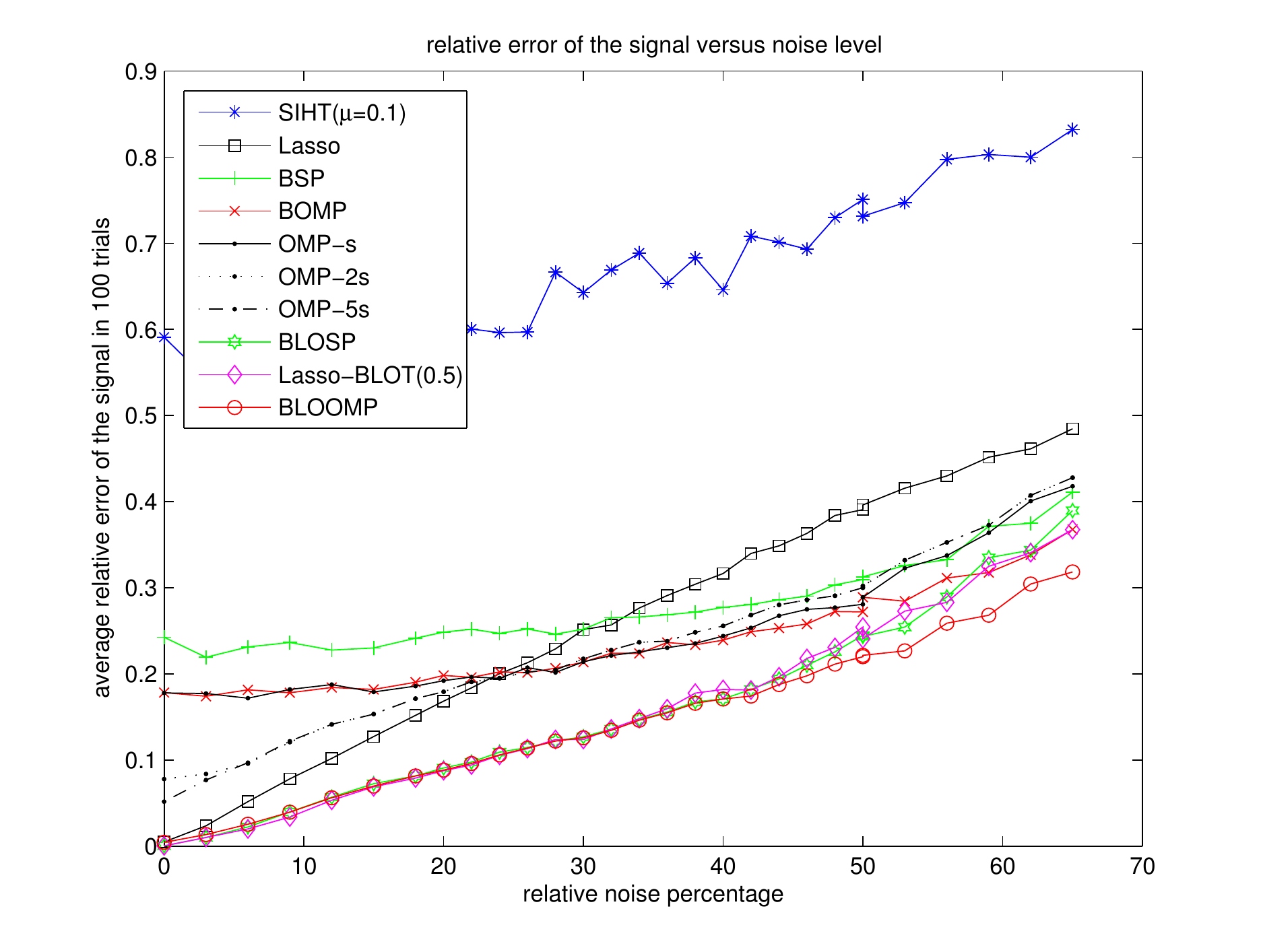}
\includegraphics[width=8cm]{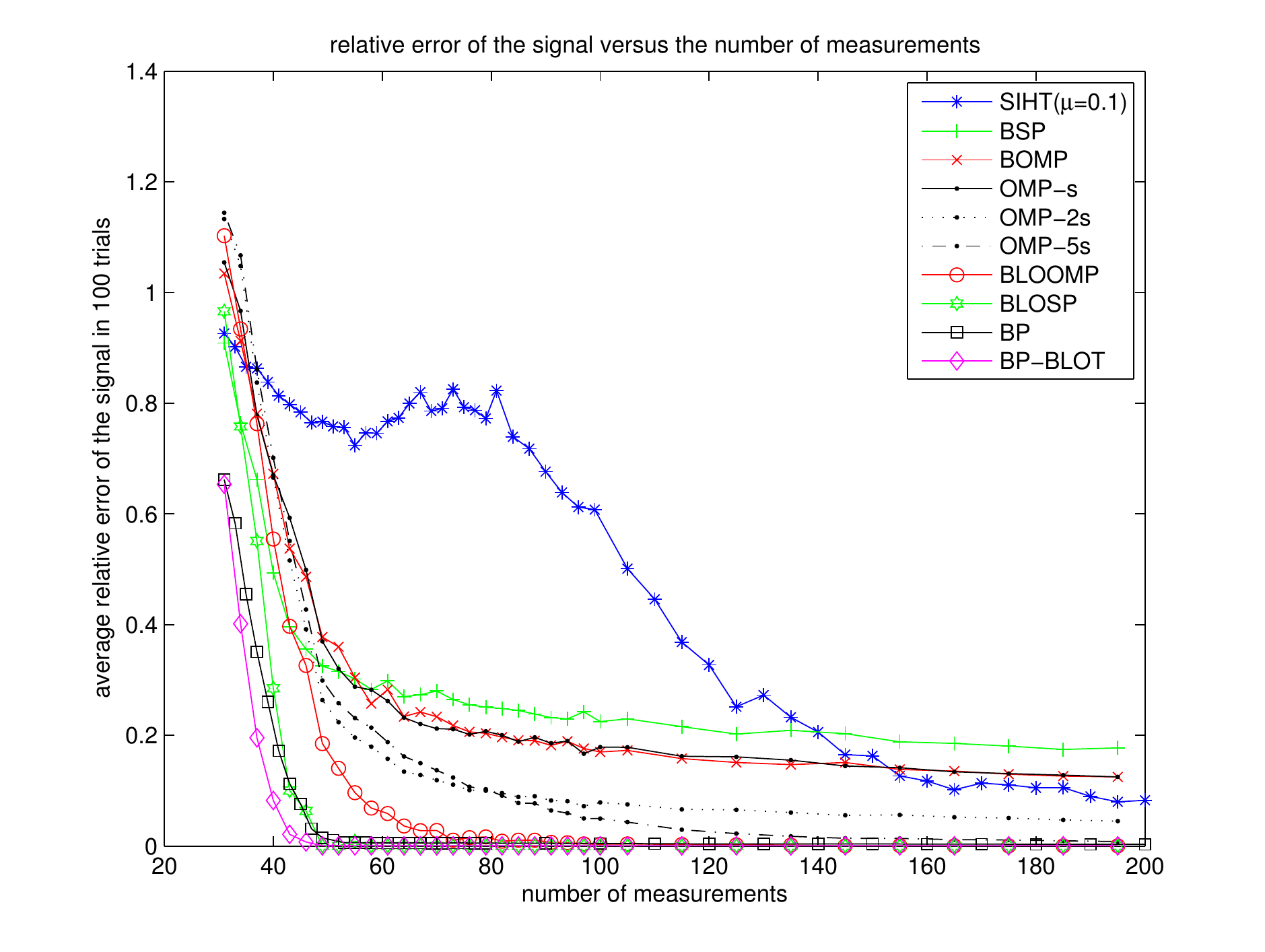}
\includegraphics[width=8cm]{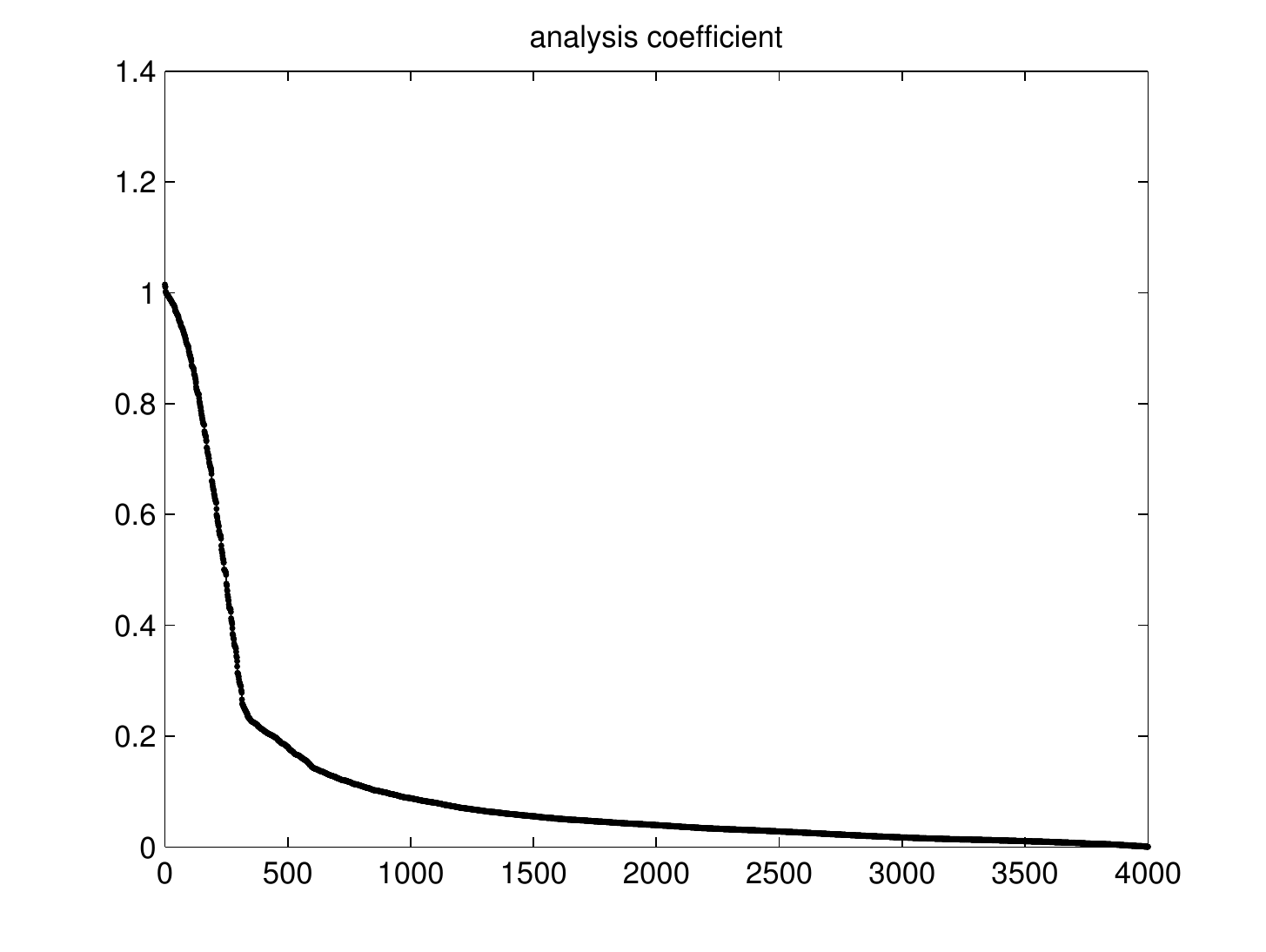}
\caption{Relative errors versus dynamic range (top left), relative noise (top right) and number of measurements (bottom left). The bottom right panel shows the magnitudes  of a realization of the coefficient vector $\bPsi^*\by$ in the descending order
of magnitude.  }
\label{fig19}
\end{figure}

The results are shown in Figure \ref{fig19}.   Consistently across the
top left, top right and bottom left panels, the smallest error is achieved by BP-BLOT and Lasso-BLOT with (\ref{10.3}) with respect to 
dynamic range (top left), noise (top right)
and number of measurements (bottom left). The latter two plots
are for dynamic range 1. 
BLOOMP and BLOSP perform among the best with respect
to dynamic range (top left) and noise (top right)
and can achieve the minimum error  with increasing number of measurements (bottom left).   
The SIHT algorithm requires much higher number of
measurements to get its error down (bottom left)
and produces the highest level of error w.r.t. dynamic
range (top left) and noise (top right). 

In Figure \ref{fig19}, we include the performance curves of 
OMP with various sparsities (s, 2s, 5s) as well as 
the standard BP/Lasso without BLOT.  Not surprisingly, the relative error of OMP reconstruction
decreases with increasing sparsity. 

It is noteworthy that the BLOT technique reduces the BP/Lasso reconstruction errors:
BP without BLOT produces $0.5\%$ relative error with respect to 
dynamic range (top left, not clearly visible)
while  BP-BLOT, along with BLOSP and BLOOMP, produces  $10^{-16}$ relative error.  Moreover,  Lasso with the optimized parameter (\ref{10.3}) but without BLOT produces significantly higher errors 
than Lasso-BLOT, BLOOMP and BLOSP (top right).

\section{ Conclusions and Discussions}\label{sec7}

We have developed and tested various algorithms for sparse recovery  with
highly coherent sensing matrices arising in discretization of 
 imaging problems in continuum 
such as radar  and medical imaging when the grid spacing
is below the Rayleigh threshold \cite{Don1}.

We have introduced two essential  techniques to deal with
unresolved grids: band exclusion and local optimization. We have embedded these techniques in various CS algorithms and 
performed systematic tests on them. When embedded in OMP, both  BE and
LO steps manifest their advantage in dealing with larger
dynamic range. When embedded in SP, CoSaMP, IHT,  BP and
Lasso the effects are more dramatic. 

We have studied these modified  algorithms from four
performance metrics: dynamic range, noise stability, 
sparsity  and resolution. 
With respect to  the first two metrics (dynamic range and noise stability),
BLOOMP is the best performer. 
With respect to  sparsity, BLOOMP is the best
performer for high dynamic range while for dynamic range
near unity BP-BLOT and Lasso-BLOT with the optimal regularization parameter have the best performance. 
BP-BLOT also has the highest resolving power
up to certain dynamic range. Lasso-BLOT's performance, however,
is sensitive to the choice of regularization parameter

One of the 
most  surprising  attributes of BLOOMP is 
improved  performance  with respect to sparsity
at {\em larger}  dynamic range and low noise. 

 The algorithms
BLOSP, BLOCoSaMP and  BLOIHT  are good alternatives  to 
 BLOOMP and BP/Lasso-BLOT: they are faster
than both BLOOMP and BP/Lasso-BLOT  and shares, to a lesser degree, BLOOMP's desirable attribute with respect to dynamic range.   

Comparisons  with existing algorithms (SIHT and frame-adapted
BP) demonstrate the superiority of BLO-enhanced   algorithms    for reconstruction of sparse objects separated above the Rayleigh length. 

Finally to add to the debate of analysis versus synthesis  \cite{EMR},   the performance of BLO-based algorithms
for sparse, widely separated objects 
are independent of the refinement factors $F$ representing  redundancy, and,  since the discretization error decreases
with $F$, the reconstruction errors of the BLO-based synthesis 
methods  also decrease with $F$
in stark contrast to the examples presented in \cite{EMR}
which show that the synthesis approach degrades
with redundancy.

\commentout{
Finally, reconstruction with unresolved grids
is the first step toward superresolution as analyzed  in \cite{Don1},
even though our proposed algorithms can not resolve
length below 1 RL in terms of  the Bottleneck distance
which is a worst case metric. 
As Donoho says \cite{Don1}, 
``We do not exhibit a practical method for achieving stable recovery, but instead exhibit inequalities which show that, in the sense of the theory of optimal recovery, the object admits of stable reconstruction."
}

\bibliographystyle{plain}	
\bibliography{myref}		
 
\end{document}